\documentclass[
3p,
]{elsarticle}

\usepackage{pgfplots}
	\pgfplotsset{compat=1.12}
\usepackage{tikz}
 	\usetikzlibrary{decorations.markings}
 	\usetikzlibrary{calc}
\usepackage{amsmath,amsthm,amssymb}
\usepackage{hyperref}
	\hypersetup{pdfauthor=author}
\usepackage[capitalize]{cleveref}
	\crefname{equation}{}{}
	\crefformat{section}{\textsection#2#1#3}
	\crefformat{appendix}{\textsection#2#1#3}

\theoremstyle{plain}
	\newtheorem{mtheorem}{Theorem}
	
	\newtheorem{theorem}{Theorem}[section]
	\newtheorem{lemma}[theorem]{Lemma}
	
	\newtheorem{corollary}[theorem]{Corollary}

\theoremstyle{definition}

	\newtheorem{remark}[theorem]{Remark}
	\newtheorem{example}[theorem]{Example}
	
	\newenvironment{eg}{\begin{footnotesize} \begin{example}}{\end{example}\end{footnotesize}}

\newcommand{\N}{\mathbb{N}}
\newcommand{\Z}{\mathbb{Z}}
\newcommand{\R}{\mathbb{R}}
\newcommand{\C}{\mathbb{C}}
\newcommand{\T}{\mathbb{T}}

\newcommand{\rL}{\mathrm{L}}
\newcommand{\rR}{\mathrm{R}}

\newcommand{\cA}{\mathcal{A}}
\newcommand{\cF}{\mathcal{F}}
\newcommand{\cB}{\mathcal{B}}
\newcommand{\cK}{\mathcal{K}}
\newcommand{\cH}{\mathcal{H}}

\newcommand{\sgn}{\mathrm{sgn}\,}

\newcommand{\inner}[1]{\left \langle #1 \right \rangle}
\newcommand{\wn}{\mathrm{wn}}
\renewcommand{\Re}{\mathrm{Re}\,}
\renewcommand{\Im}{\mathrm{Im}\,}
\newcommand{\ind}{\mathrm{ind}\,}

\newcommand{\ess}{\sigma_{\mathrm{ess}}}

\newcommand{\textbi}[1]{\textit{\textbf{#1}}}

\begin{document}

\begin{frontmatter}

\title{
A Constructive Approach to Topological Invariants for One-dimensional Strictly Local Operators
}

\author[Shinshu]{Yohei Tanaka\corref{corresponding}}
\ead{20hs602a@shinshu-u.ac.jp}

\cortext[corresponding]{Corresponding author}

\address[Shinshu]{Division of Mathematics and Physics, Faculty of Engineering, Shinshu University, Wakasato, Nagano 380-8553, Japan}

\begin{abstract}
In this paper we shall focus on one-dimensional strictly local operators, the notion of which naturally arises in the context of discrete-time quantum walks on the one-dimensional integer lattice $\Z.$ In particular, we give an elementary constructive approach to the following two topological invariants associated with such operators: Fredholm index and essential spectrum. As a direct application, we shall explicitly compute and fully classify these topological invariants for a well-known quantum walk model.
\end{abstract}

\end{frontmatter}

\section{Introduction}

The underlying Hilbert space of this manuscript is $\ell^2(\Z, \C^n)$ of square-summable $\C^n$-valued sequences indexed by the set $\Z$ of all integers. With the obvious orthogonal decomposition $\ell^2(\Z, \C^n) = \bigoplus_{j=1}^n \ell^2(\Z, \C)$ in mind, we shall consider the following \textit{finite} sum of operators;
\begin{equation}
\label{equation2: characterisation of strict locality}
A
=
\sum^k_{y=-k}
\begin{pmatrix}
a_{11}(y, \cdot) L^{y} & \dots & a_{1n}(y, \cdot) L^{y} \\
\vdots & \ddots& \vdots \\
a_{n1}(y, \cdot) L^{y} & \dots & a_{nn}(y, \cdot) L^{y} \\
\end{pmatrix}
=
\begin{pmatrix}
\sum^k_{y=-k} a_{11}(y, \cdot) L^{y} & \dots & \sum^k_{y=-k} a_{1n}(y, \cdot) L^{y} \\
\vdots & \ddots& \vdots \\
\sum^k_{y=-k} a_{n1}(y, \cdot) L^{y} & \dots & \sum^k_{y=-k} a_{nn}(y, \cdot) L^{y} \\
\end{pmatrix},
\end{equation}
where $L$ denotes the bilateral left-shift operator on $\ell^2(\Z, \C)$ (see Equation \cref{definition: bileteral shift} for definition) and where each $a_{ij}(y, \cdot) = (a_{ij}(y, x))_{x \in \Z}$ is an arbitrary bounded $\C$-valued sequence viewed as a multiplication operator on $\ell^2(\Z, \C) = \bigoplus_{x \in \Z} \C.$ An operator the form \cref{equation2: characterisation of strict locality} is known as a (one-dimensional) \textbi{strictly local operator} \cite[\textsection 1.2]{Cedzich-Geib-Stahl-Velazquez-Werner-Werner-2018}. Such an operator naturally arises, for example, in the context of $n$-state quantum walks defined on the integer lattice $\Z,$ where we regard $\ell^2(\Z, \C^n)$ as the state Hilbert space of the walker. The purpose of this paper is to prove the following general theorem;

\begin{mtheorem}
\label{theorem: two-phase case}
Let $A$ be a strictly local operator of the form \cref{equation2: characterisation of strict locality} with the property that the following two-sided limits exist:
\begin{equation}
\label{equation: two-phase assumptions}
\tag{A1}
a_{ij}(y, \pm \infty) := \lim_{x \to \pm \infty} a_{ij}(y,x) \in \C, \qquad i,j = 1, \dots, n ,\ -k \leq y \leq k.
\end{equation}
Let
\begin{align}
\label{equation: definition of Apm}
\tag{A2}
A(\pm \infty)
&:=
\begin{pmatrix}
\sum^k_{y=-k} a_{11}(y, \pm \infty) L^{y} & \dots & \sum^k_{y=-k} a_{1n}(y, \pm \infty) L^{y} \\
\vdots & \ddots& \vdots \\
\sum^k_{y=-k} a_{n1}(y, \pm \infty) L^{y} & \dots & \sum^k_{y=-k} a_{nn}(y, \pm \infty) L^{y} \\
\end{pmatrix}, \\
\label{equation: definition of hatA}
\tag{A3}
\hat{A}(z,\pm \infty)
&:=
\begin{pmatrix}
\sum^k_{y=-k} a_{11}(y, \pm \infty) z^{y} & \dots & \sum^k_{y=-k} a_{1n}(y, \pm \infty) z^{y} \\
\vdots & \ddots& \vdots \\
\sum^k_{y=-k} a_{n1}(y, \pm \infty) z^{y} & \dots & \sum^k_{y=-k} a_{nn}(y, \pm \infty) z^{y} \\
\end{pmatrix}, \qquad z \in \T.
\end{align}
Then the following assertions hold true:
\begin{enumerate}[(i)]
\item We have that $A$ is Fredholm if and only if $\T \ni z \longmapsto \det \hat{A}(z,\star) \in \C$ is nowhere vanishing on $\T$ for each $\star = \pm \infty.$ In this case, the Fredholm index of $A$ is given by
\begin{equation}
\label{equation: bulk-edge correspondence}
\tag{A4}
\ind(A) = \wn \left(\det \hat{A}(\cdot,+\infty) \right) - \wn \left(\det \hat{A}(\cdot, -\infty) \right),
\end{equation}
where $\wn \left(\det \hat{A}(\cdot, \star) \right)$ denotes the winding number of the function $\T \ni z \longmapsto \det \hat{A}(z,\star) \in \C$ with respect to the origin.

\item The essential spectrum of $A$ is given by
\begin{align}
\label{equation1: essential spectrum of two-phase operator}
\tag{A5}
&\ess(A) = \ess(A(+ \infty)) \cup \ess(A(- \infty)), \\
\label{equation2: essential spectrum of two-phase operator}
\tag{A6}
&\ess(A(\star)) = \bigcup_{z \in \T} \sigma \left(\hat{A}(z,\star)\right), \qquad \star = \pm \infty.
\end{align}
\end{enumerate}
\end{mtheorem}

Note that the Fredholm index and essential spectrum can be viewed as topological invariants in the sense that these are stable under compact perturbations as is well-known. \cref{theorem: two-phase case} allows us to explicitly compute these topological invariants for strictly local operators, both of which depend only on the asymptotic values \cref{equation: two-phase assumptions}. As we shall see in this article, \cref{theorem: two-phase case}(i) can be viewed as an abstract version of the one-dimensional bulk-boundary correspondence (see, for example, \cite[Corollary 4.3]{Cedzich-Geib-Grunbaum-Stahl-Velazquez-Werner-Werner-2018}). In the setting of discrete-time quantum walks, a standard approach to \cref{theorem: two-phase case}(ii) involves the use of the discrete Fourier transform \cref{equation: definition of hatA} and Weyl's criterion for the essential spectrum; see, for example, \cite[Lemma 3.3]{Fuda-Funakawa-Suzuki-2017}, where the essential spectrum of the self-adjoint discriminant operator is determined via a certain spectral mapping theorem for quantum walks \cite{Segawa-Suzuki-2016,Segawa-Suzuki-2019}. Weyl's criterion is capable of dealing with a wide range of perturbations that are not necessarily compact (see, for example, \cite{Sasaki-Suzuki-2017}), but its usage is obviously restricted to normal operators, unlike \cref{theorem: two-phase case}(ii).

This paper is organised as follows. Proof of \cref{theorem: two-phase case} is given in \cref{section: analysis of strictly local operators}. In particular, our proof of the index formula \cref{equation: bulk-edge correspondence} is entirely motivated by \cite{Matsuzawa-2020}, where the special case $n=1$ of the index formula is established by making use of the notion of Toeplitz operators with $\C$-valued symbols. We show that \cref{theorem: two-phase case}(ii) can be proved without relying on Weyl's criterion, if we allow Toeplitz operators to have $\C^{n \times n}$-valued symbols. The paper concludes with \cref{section: examples}, where we explicitly compute and fully classify the \textbi{Witten index} of a certain quantum walk model on $\ell^2(\Z,\C^2)$ with the aid of \cref{theorem: two-phase case}(i). This is the one-dimensional split-step quantum walk model considered in \cite{Fuda-Funakawa-Suzuki-2017,Fuda-Funakawa-Suzuki-2018,Fuda-Funakawa-Suzuki-2019,Suzuki-Tanaka-2019,Matsuzawa-2020} with a modification that all parameters of the model depend freely on $\Z.$ It is shown that this seemingly minor modification leads to the new index formula taking values from $\{-2,-1,0,+1,+2\},$ where the indices $\pm 2$ do not appear in the existing literature \cite{Suzuki-Tanaka-2019,Matsuzawa-2020}. This result turns out to be significant improvement, since the Witten index gives a lower bound for the number of \textit{topologically protected bound states} in the sense of \cref{section: topologically protected bound states}. As a direct application of \cref{theorem: two-phase case}(ii), we shall also compute the essential spectrum of the associated time-evolution operator.

\section*{Acknowledgements}

The author would like to thank the members of the Shinshu Mathematical Physics Group for extremely useful comments and discussions. In particular, his sincerely thanks goes to K.~Wada for pointing out that \cite[Theorem 8]{Suzuki-Tanaka-2019} can be generalised to \cref{lemma: wada decomposition}. He would also like to thank M.~Seki and K.~Asahara for carefully checking the manuscript. This research was supported by the Ministry of Education, Science, Sports and Culture, Grant-in-Aid for JSPS Fellows, 20J22684, 2020. This work was also partially supported by the Research Institute for Mathematical Sciences, an International Joint Usage/Research Center located in Kyoto University.

\section{Analysis of Strictly Local Operators}
\label{section: analysis of strictly local operators}

\subsection{Notation and terminology}
By operators we shall always mean everywhere-defined bounded linear operators between Banach spaces throughout this paper. An operator $A$ on a Hilbert space $\cH$ is said to be \textbi{Fredholm}, if $\ker A, \ker A^*$ are finite-dimensional and if $A$ has a closed range. Given such $A,$ we define the \textbi{Fredholm index} of $A$ by $\ind(A) := \dim \ker A - \dim \ker A^*.$ It is well-known that the Fredholm index is invariant under compact perturbations. That is, given an operator $A$ on $\cH$ and a compact operator $K$ on $\cH,$ we have that $A$ is Fredholm if and only if so is $A + K,$ and in this case $\ind(A) = \ind(A + K).$ The (Fredholm) \textbi{essential spectrum} of an operator $A$ on $\cH$ is defined as the set $\ess(A)$ of all $\lambda \in \C,$ such that $A - \lambda$ fails to be Fredholm. Note that  $\ess(A)$ is also stable under compact perturbations.

The Hilbert space of all square-summable $\C$-valued sequences $\Psi = (\Psi(x))_{x \in \Z}$ is denoted by the shorthand $\ell^2(\Z) := \ell^2(\Z,\C).$ We have a natural orthogonal decomposition $\ell^2(\Z) = \ell^2_\rL(\Z) \oplus \ell^2_\rR(\Z),$ where
\[
\ell^2_{\rL}(\Z) := \{\Psi \in \ell^2(\Z) \mid \Psi(x) = 0 \,\, \forall x \geq 0 \}, \qquad
\ell^2_\rR(\Z) := \{\Psi \in \ell^2(\Z) \mid \Psi(x) = 0 \,\, \forall x < 0\}.
\]
The orthogonal projections of $\ell^2(\Z)$ onto the above subspaces shall be denoted by $P_{\rL}$ and $P_{\rR} = 1 - P_{\rL}$ respectively. For each $\sharp  \in \{\rL,\rR\},$ the orthogonal projection $P_{\sharp}$ can be written as $P_{\sharp} = \iota_{\sharp} \iota_{\sharp}^*,$ where $\iota_{\sharp} : \ell^2_{\sharp}(\Z) \hookrightarrow \ell^2(\Z)$ is the inclusion mapping. The left-shift operator $L$ on $\ell^2(\Z)$ is defined by
\begin{equation}
\label{definition: bileteral shift}
L \Psi := \Psi(\cdot + 1), \qquad \Psi \in \ell^2(\Z).
\end{equation}

Let $n \in \N$ be fixed throughout the current section. Any operator $A$ on $\ell^2(\Z, \C^n) := \bigoplus_{j=1}^n \ell^2(\Z)$ admits the following unique block-operator matrix representation;
\begin{equation}
\label{equation: standard representation of A}
A =
\begin{pmatrix}
A_{11} & \dots & A_{1n} \\
\vdots & \ddots& \vdots \\
A_{n1} & \dots & A_{nn} \\
\end{pmatrix}_{\bigoplus_{j=1}^n \ell^2(\Z)},
\end{equation}
where each $A_{ij}$ is an operator on $\ell^2(\Z).$ We shall agree to use the shorthand $A = (A_{ij})$ to mean that \cref{equation: standard representation of A} holds true. With this representation of $A$ in mind, for each $\sharp  = \rL, \rR,$ we define the following operator on $\ell^2_{\sharp}(\Z, \C^n) := \bigoplus_{j=1}^n \ell^2_{\sharp}(\Z);$
\begin{equation}
\label{equation: matrix repesentation of contraction}
A_{\sharp} :=
\begin{pmatrix}
\iota_{\sharp}^* A_{11}\iota_{\sharp} & \dots & \iota_{\sharp}^* A_{1n}\iota_{\sharp} \\
\vdots & \ddots& \vdots \\
\iota_{\sharp}^* A_{n1}\iota_{\sharp} & \dots & \iota_{\sharp}^* A_{nn}\iota_{\sharp} \\
\end{pmatrix}_{\bigoplus_{j=1}^n \ell_\sharp^2(\Z)}.
\end{equation}

\subsection{Strictly local operators}

\begin{lemma}
\label{lemma: characterisation of strict locality}
Let $(\delta_x)_{x \in \Z}$ be the standard complete orthonormal basis for $\ell^2(\Z),$ and let $A$ be an operator on $\ell^2(\Z, \C^n)$ with the block-operator matrix representation \cref{equation: standard representation of A}. Then the following are equivalent:
\begin{enumerate}[(i)]
\item For each $i,j \in \{1, \dots, n\},$ the operator $A_{ij}$ is a finite sum of the form $A_{ij} = \sum^k_{y=-k} a_{ij}(y, \cdot)L^y$ for some $\C$-valued sequences $a_{ij}(y, \cdot) = (a_{ij}(y, x))_{x \in \Z},$ where $-k \leq y \leq k,$ viewed as multiplication operators on $\ell^2(\Z) = \bigoplus_{x \in \Z} \C.$

\item There exists a large enough positive integer $k,$ such that for each $x \in \Z$ and for each $i,j \in \{1, \dots, n\},$ the vector $A_{ij} \delta_x \in \ell^2(\Z)$ belongs to the linear span of the finite set $\{ \delta_{x-y} \mid -k \leq y \leq k \}.$
\end{enumerate}
\end{lemma}
Following \cite[\textsection 1.2]{Cedzich-Geib-Stahl-Velazquez-Werner-Werner-2018}, any operator $A$ satisfying the above equivalent conditions shall be referred to as a \textbi{strictly local operator} from here on. It follows from (i) that such $A$ admits a block-matrix representation of the form \cref{equation2: characterisation of strict locality}.
\begin{proof}
It is obvious that (i) implies (ii), since $L^y \delta_x  = \delta_{x - y}$ for each $x,y \in \Z.$ This equality shall be repeatedly used throughout the current section. To prove the converse, let (ii) hold true, and let $i,j$ be both fixed. For each $x \in \Z,$ there exist finitely many scalars $a'_{ij}(y, x) \in \C,$ where $-k \leq y \leq k,$ such that
\begin{equation}
\label{equation3: characterisation of strict locality}
A_{ij} \delta_x = \sum_{y=-k}^k a'_{ij}(y, x) \delta_{x-y}.
\end{equation}
Note that $a'_{ij}(y, \cdot) = (a'_{ij}(y, x))_{x \in \Z}$ is a bounded sequence for $-k \leq y \leq k;$
\[
|a'_{ij}(y, x)| = |\inner{\delta_{x - y}, A_{ij}\delta_{x}}_{\ell^2(\Z)}| \leq \|A_{ij}\| \|\delta_{x-y}\|_{\ell^2(\Z)}  \|\delta_{x}\|_{\ell^2(\Z)} \leq \|A_{ij}\|, \qquad x \in \Z.
\]
Let $a_{ij}(y,x) := a'_{ij}(y, x + y)$ for each $x, y.$ Then we obtain the following equality for each $x \in \Z;$
\[
\sum^k_{y=-k} a_{ij}(y, \cdot) L^{y} \delta_x
= \sum^k_{y=-k} a_{ij}(y, \cdot) \delta_{x-y}
= \sum^k_{y=-k} a_{ij}(y, x-y) \delta_{x-y}
= \sum^k_{y=-k} a'_{ij}(y, x) \delta_{x-y}
= A_{ij} \delta_x,
\]
where the last equality follows from \cref{equation3: characterisation of strict locality}. That is, (i) holds true.
\end{proof}

\begin{corollary}
\label{corollary: properties of strictly local operators}
If $A$ is a strictly local operator on $\ell^2(\Z, \C^n),$ then the difference $A - A_{\rL} \oplus A_{\rR}$ is finite-rank. Moreover, the following assertions hold true:
\begin{enumerate}[(i)]
\item The operator $A$ is Fredholm if and only if $A_\rL, A_{\rR}$ are both Fredholm. In this case, we have
\begin{equation}
\label{equation: fredholm index of strictly local operators}
\ind(A) = \ind(A_{\rL}) + \ind(A_{\rR}).
\end{equation}
\item The essential spectrum of $A$ is given by
\begin{equation}
\label{equation: essential spectrum of strictly local operators}
\ess(A) = \ess(A_{\rL}) \cup \ess(A_{\rR}).
\end{equation}
\end{enumerate}
\end{corollary}
\begin{proof}
Note that $P := \bigoplus_{j=1}^n P_{\rR}$ is the orthogonal projection onto $\ell_{\rR}^2(\Z, \C^n) = \bigoplus_{j=1}^n \ell_{\rR}^2(\Z).$ We have
\[
A -  A_{\rL} \oplus A_{\rR}
= PA(1 - P) + (1-P)AP
= PA - PAP + AP - PAP
= P[P,A] + [A,P]P,
\]
where $[X,Y] := XY - YX$ denotes the commutator of two operators $X,Y.$ It remains to show that $[A,P]$ is finite-rank, where we may assume without loss of generality that $A$ is of the form \cref{equation2: characterisation of strict locality}. Since $P = \bigoplus_{j=1}^n P_{\rR}$ is a diagonal block-operator matrix, we obtain
\[
[A,P] =
\begin{pmatrix}
\left[\sum^k_{y=-k} a_{11}(y, \cdot) L^{y}, P_{\rR}\right]& \dots & \left[\sum^k_{y=-k} a_{1n}(y, \cdot) L^{y}, P_{\rR} \right] \\
\vdots & \ddots& \vdots \\
\left[\sum^k_{y=-k} a_{n1}(y, \cdot) L^{y}, P_{\rR}\right] & \dots & \left[\sum^k_{y=-k} a_{nn}(y, \cdot) L^{y}, P_{\rR}\right] \\
\end{pmatrix}.
\]
Since $[-, P_{\rR}]$ is linear with respect to the first variable, each $(i,j)$-entry of the above block-operator matrix is given by $\sum^k_{y=-k}  a_{ij}(y, \cdot) \left[ L^{y}, P_{\rR} \right],$ where the commutator $\left[ L^{y}, P_{\rR} \right]$ is finite-rank for $-k \leq y \leq k.$ It follows that $A - A_{\rL} \oplus A_{\rR}$ is finite-rank, and so the remaining assertions immediately follow.
\end{proof}

Note that a strictly local operator of the form \cref{equation2: characterisation of strict locality} has the simplest formula, if each sequence $a_{ij}(y,\cdot)$ is constant. Such an operator admits the following characterisation;

\begin{lemma}
\label{lemma: characterisation of uniformity}
Let $A$ be an operator on $\ell^2(\Z)$ with the block-operator matrix representation \cref{equation: standard representation of A}. Then the following are equivalent:
\begin{enumerate}[(i)]
\item For each $i,j \in \{1, \dots, n\},$ the operator $A_{ij}$ is a finite sum of the form $A_{ij} = \sum^k_{y=-k} a_{ij}(y) L^y$ for some complex numbers $a_{ij}(y),$ where $-k \leq y \leq k.$ 

\item The operator $A$ is strictly local and $[A_{ij}, L^x] = 0$ for each $x \in \Z$ and each $i,j \in \{1, \dots, n\}.$
\end{enumerate}
\end{lemma}
The operator $A$ is said to be \textbi{uniform}, if it satisfies the above equivalent conditions. It follows from (i) that such $A$ admits a block-matrix representation of the following form;
\begin{equation}
\label{equation: characterisation of uniformity}
A
=
\sum^k_{y=-k}
\begin{pmatrix}
a_{11}(y) L^{y} & \dots & a_{1n}(y) L^{y} \\
\vdots & \ddots& \vdots \\
a_{n1}(y) L^{y} & \dots & a_{nn}(y) L^{y} \\
\end{pmatrix}
=
\begin{pmatrix}
\sum^k_{y=-k} a_{11}(y) L^{y} & \dots & \sum^k_{y=-k} a_{1n}(y) L^{y} \\
\vdots & \ddots& \vdots \\
\sum^k_{y=-k} a_{n1}(y) L^{y} & \dots & \sum^k_{y=-k} a_{nn}(y) L^{y} \\
\end{pmatrix}.
\end{equation}

\begin{proof}
It is obvious that (i) implies (ii). To prove the converse, let $A = (A_{ij})$ be strictly local, and let $[A_{ij}, L^x] = 0$ for each $x \in \Z$ and for each $i,j \in \{1, \dots, n\}.$ It follows from \cref{lemma: characterisation of strict locality}(i) that for each $i,j \in \{1, \dots, n\}$ we have $A_{ij} = \sum^k_{y=-k} a_{ij}(y,\cdot) L^y.$ It remains to show that the sequence $a_{ij}(y, \cdot) = (a_{ij}(y, x))_{x \in \Z}$ is constant for a fixed pair $i,j \in \{1, \dots, n\}.$ Since $A_{ij}\delta_x = \sum^k_{y=-k} a_{ij}(y,x-y) \delta_{x-y}$ for each $x \in \Z,$ we get
\[
A_{ij}\delta_x
= A_{ij} L^{-x} \delta_0
= L^{-x} A_{ij}  \delta_0
= L^{-x} \left(\sum^k_{y=-k} a_{ij}(y,-y) \delta_{-y} \right)
= \sum^k_{y=-k} a_{ij}(y,-y) \delta_{x-y}.
\]
It follows that $a_{ij}(y,x-y) = a_{ij}(y,-y)$ for each $x \in \Z$ and for each $y \in \{-k, \dots, k\}.$ The claim follows.
\end{proof}

\subsection{Uniform operators}

The following result is one of the main theorems of the current section;
\begin{theorem}
\label{theorem: uniform operators}
Let $A$ be a uniform operator on $\cH = \ell^2(\Z, \C^n)$ of the form \cref{equation: characterisation of uniformity}, and let
\begin{equation}
\label{equation: definition of f}
\hat{A}(z) :=
\begin{pmatrix}
\sum^k_{y=-k} a_{11}(y) z^{y} & \dots & \sum^k_{y=-k} a_{1n}(y) z^{y} \\
\vdots & \ddots& \vdots \\
\sum^k_{y=-k} a_{n1}(y) z^{y} & \dots & \sum^k_{y=-k} a_{nn}(y) z^{y} \\
\end{pmatrix}, \qquad z \in \T.
\end{equation}
Then the following assertions hold true:
\begin{enumerate}[(i)]
\item The operator $A$ is Fredholm if and only if $A_\rL, A_{\rR}$ are both Fredholm if and only if $\T \ni z \longmapsto \det \hat{A}(z) \in \C$ is nowhere vanishing on $\T.$ In this case, we have $\ind A = \ind A_\rR + \ind A_\rL = 0,$ and
\[
\ind A_{\rR} = \wn \left(\det \hat{A} \right).
\]

\item The essential spectrum of $A$ is given by
\begin{equation}
\label{equation: essential spectra of toeplitz operators}
\ess(A) = \ess(A_{\rL}) = \ess(A_{\rR})
= \bigcup_{z \in \T}
\sigma (\hat{A}(z)).
\end{equation}
\end{enumerate}
\end{theorem}

A proof of \cref{theorem: uniform operators} shall be given at the end of the current subsection. Let us first recall the notion of Toeplitz operators. Let $L^2(\T)$ be the Hilbert space of square-summable functions on the unit-circle $\T,$ where $\T$ is endowed with the normalised arc-length measure. It is well-known that $L^2(\T)$ admits the standard complete orthonormal basis $(e_x)_{x \in \Z},$ where each $e_x$ is defined by $\T \ni z \longmapsto z^x \in \C.$ The \textbi{Hardy-Hilbert space} $\mathrm{H}^2$ is the closure of the linear span of the set $\{e_x \mid x \geq 0\}.$ Let $\iota : \mathrm{H}^2 \hookrightarrow L^2(\T)$ be the inclusion mapping, and let $f \in \C(\T).$  Then the \textbi{Toeplitz operator} $T_f$ with symbol $f$ is defined by
\begin{equation}
\label{equation: definition of toeplitz operator}
T_f := \iota^* M_f \iota,
\end{equation}
where $M_f : L^2(\T) \to L^2(\T)$ is the multiplication operator by $f.$ More generally, let us consider the Banach space $C(\T, \C^{n \times n})$ of continuous matrix-valued functions on $\T.$ Given a function $F \in C(\T, \C^{n \times n})$ of the form $F(z) = (f_{ij}(z))$ for each $z \in \T,$ the Toeplitz operator with symbol $F$ is defined by
\begin{equation}
\label{equation: definition of toeplitz operator with matrix-valued symbol}
T_F :=
\begin{pmatrix}
T_{f_{11}} & \dots & T_{f_{1n}} \\
\vdots & \ddots & \vdots\\
T_{f_{n1}} & \dots & T_{f_{nn}}
\end{pmatrix}_{\bigoplus_{j=1}^n \mathrm{H}^2}.
\end{equation}
The following result is standard;


\begin{theorem}
\label{theorem: Gohberg-Krein}
Let $F \in C(\T, \C^{n \times n})$ be a matrix-valued function of the form $F(\cdot) = (f_{ij}(\cdot)),$ and let $T_F$ be the corresponding Toeplitz operator given by \cref{equation: definition of toeplitz operator with matrix-valued symbol}. Then the following assertions hold true:
\begin{enumerate}[(i)]
\item The Toeplitz operator $T_F$ is Fredholm if and only if $\T \ni z \longmapsto \det F(z) \in \C$ is nowhere vanishing on $\T.$ In this case,
\begin{equation}
\label{equation: fredholm index of matrix toeplitz operator}
\ind T_F = - \wn(\det F).
\end{equation}

\item The essential spectrum of $T_F$ is given by
\begin{equation}
\label{equation: essential spectrum of matrix toeplitz operator}
\ess(T_F) = \bigcup_{z \in \T} \sigma(F(z)).
\end{equation}
\end{enumerate}
\end{theorem}
\begin{proof}
Note that (i) is the celebrated theorem of Gohberg-Krein (see, for example, \cite[Theorem 3.3]{Murphy-2006}). It remains to prove (ii). Let $\cB_n(\mathrm{H}^2) := \cB \left(\bigoplus_{j=1}^n \mathrm{H}^2\right)$ be the $C^*$-algebra of operators on $\bigoplus_{j=1}^n \mathrm{H}^2,$ and let $\cK_n(\mathrm{H}^2) := \cK \left(\bigoplus_{j=1}^n \mathrm{H}^2\right)$ be the ideal of compact operators on $\bigoplus_{j=1}^n \mathrm{H}^2.$ Let $\cA_n(\mathrm{H}^2)$ be the closed $*$-subalgebra of $\cB_n(\mathrm{H}^2)$ generated by $\{T_F \mid F \in  C(\T, \C^{n \times n})\}.$ It is a well-known result that the following mapping is $*$-isomorphic (see, for example, \cite[\textsection 1]{Douglas-1973});
\[
C(\T, \C^{n \times n}) \ni F \longmapsto [T_F] \in \cA_n(\mathrm{H}^2) / \cK_n(\mathrm{H}^2).
\]
That is, for each $F \in C(\T, \C^{n \times n})$ we have that $F$ is invertible in $C(\T, \C^{n \times n})$ if and only if $[T_F]$ is invertible in the Calkyin algebra $\cB_n(\mathrm{H}^2) / \cK_n(\mathrm{H}^2).$ The equality \cref{equation: essential spectrum of matrix toeplitz operator} follows.
\end{proof}

Let us consider two unitary operators $\cF_{\rL} :  \mathrm{H}^2 \to \ell^2_{\rL}(\Z)$ and $\cF_{\rR} : \mathrm{H}^2 \to \ell^2_{\rR}(\Z)$ defined respectively by
\[
\cF_{\rL} e_x := \delta_{-x-1}, \qquad \cF_{\rR}e_x  := \delta_x, \qquad x \geq 0,
\]
where $(\delta_x)_{x \in \Z}, (e_x)_{x \in \Z}$ are the standard bases of $\ell^2(\Z), \mathrm{H}^2$ respectively.

\begin{lemma}
\label{lemma: fourier transform of compression is toeplitz operator}
Let $A$ be a uniform operator on $\cH = \ell^2(\Z, \C^n)$ of the form \cref{equation: characterisation of uniformity}, and let $\hat{A}$ be given by \cref{equation: definition of f}. Then
\begin{equation}
\label{equation: essential spectra of uniform operators}
\left(\bigoplus_{j=1}^n \cF^*_{\sharp} \right) A_{\sharp} \left(\bigoplus_{j=1}^n \cF_{\sharp} \right)
=
\begin{cases}
T_{\hat{A}}, & \sharp = \rL, \\
T_{\hat{A}(-^{*})}, & \sharp = \rR,
\end{cases}
\end{equation}
where $\hat{A}(-^{*})$ denote the matrix-valued function $\T \ni z \longmapsto \hat{A}(z^{*}) \in \C^{n \times n}.$
\end{lemma}
\begin{proof}
Note first that the inverses of $\cF_{\rL}, \cF_{\rR}$ are given respectively by $\cF_{\rL}^{-1} \delta_x = \cF_{\rL}^{*} \delta_x = e_{-x-1}$ for each $x < 0$ and $\cF_{\rR}^{-1} \delta_x = \cF_{\rR}^{*} \delta_x = e_x$ for each $x \geq 0.$ Let us first prove the following non-trivial equalities:
\begin{equation}
\label{equaton: fourier transform of compression}
T_{e_y} = \cF_{\rL}^{*}\iota_{\rL}^* L^{y} \iota_{\rL} \cF_{\rL} =  \cF_{\rR}^{*}\iota_{\rR}^* L^{-y} \iota_{\rR} \cF_{\rR} , \qquad y \in \Z.
\end{equation}
Note that for each $y \geq 0$ each $x \geq 0$ we have
\begin{align*}
&T_{e_y} e_x
= \iota^* M_{e_y} \iota e_x
= \iota^* M_{e_y} e_x
= \iota^* e_{x+y}
= e_{x+y}, \\
&\cF_{\rL}^{*} \iota_{\rL}^* L^{y} \iota_{\rL} \cF_{\rL} e_x
=\cF_{\rL}^{*} \iota_{\rL}^* L^{y} \iota_{\rL} \delta_{-x-1}
=\cF_{\rL}^{*}\iota_{\rL}^* \delta_{-x-y-1}
= e_{x+y}, \\
&\cF_{\rR}^{*} \iota_{\rR}^* L^{-y} \iota_{\rR} \cF_{\rR} e_x
= \cF_{\rR}^{*} \iota_{\rR}^* L^{-y} \iota_{\rR} \delta_x
= \cF_{\rR}^{*} \iota_{\rR}^* \delta_{x+y}
= e_{x+y}.
\end{align*}
That is, we have shown that \cref{equaton: fourier transform of compression} holds true for any $y \geq 0.$ On the other hand, if $y < 0,$ then $-y > 0,$ and so
\[
T_{e_y} = (T_{e_{-y}})^*
= (\cF_{\rL}^{*} \iota_{\rL}^* L^{-y} \iota_{\rL} \cF_{\rL})^*
=
(\cF_{\rR}^{*} \iota_{\rR}^* L^{y} \iota_{\rR} \cF_{\rR})^* , \qquad y < 0.
\]
That is, \cref{equaton: fourier transform of compression} holds true for any $y \in \Z.$ Let
\[
f_{ij}(z) 
:=  \sum^k_{y=-k} a_{ij}(y) z^{y}
=  \sum^k_{y=-k} a_{ij}(y) e_{y}(z), \qquad z \in \T.
\]
Then the block-operator matrix representation of $A$ is given by
\[
A =
\begin{pmatrix}
\sum^k_{y=-k} a_{11}(y) L^{y} & \dots & \sum^k_{y=-k} a_{1n}(y) L^{y} \\
\vdots & \ddots& \vdots \\
\sum^k_{y=-k} a_{n1}(y) L^{y} & \dots & \sum^k_{y=-k} a_{nn}(y) L^{y} \\
\end{pmatrix}
=
\begin{pmatrix}
f_{11}(L) & \dots & f_{1n}(L) \\
\vdots & \ddots &  \vdots\\
f_{n1}(L) & \dots & f_{nn}(L)
\end{pmatrix}.
\]
With the representation \cref{equation: matrix repesentation of contraction} in mind, we obtain
\[
\left(\bigoplus_{j=1}^n \cF_{\sharp}^* \right) A_{\sharp} \left(\bigoplus_{j=1}^n \cF_{\sharp} \right)
 =
\begin{pmatrix}
\cF_{\sharp}^* \iota^*_{\sharp}  f_{11}(L) \iota_{\sharp} \cF_{\sharp} & \dots &  \cF_{\sharp}^* \iota^*_{\sharp}  f_{1n}(L)\iota_{\sharp}\cF_{\sharp}  \\
\vdots & \ddots & \vdots \\
\cF_{\sharp}^* \iota^*_{\sharp} f_{n1}(L) \iota_{\sharp} \cF_{\sharp}& \dots &  \cF_{\sharp}^* \iota^*_{\sharp} f_{nn}(L)\iota_{\sharp} \cF_{\sharp}
\end{pmatrix}_{\bigoplus_{j=1}^n \mathrm{H}^2}, \qquad \sharp =\rL, \rR,
\]
where \cref{equaton: fourier transform of compression} gives the following equalities for each $i,j \in \{1, \dots, n\}:$
\[
\cF_{\sharp}^* \iota^*_{\sharp} f_{ij}(L) \iota_{\sharp} \cF_{\sharp}
= \sum^k_{y=-k} a_{ij}(y) (\cF_{\sharp}^* \iota_{\sharp}^*  L^{y} \iota_{\sharp} \cF_{\sharp})
=
\begin{cases}
\sum^k_{y=-k} a_{ij}(y) T_{e_y} = T_{f_{ij}}, & \sharp = \rL, \\
\sum^k_{y=-k} a_{ij}(y) T_{e_y}^* = T_{f_{ij}(-^{*})}, &\sharp = \rR.
\end{cases}
\]
It follows that \cref{equation: essential spectra of uniform operators} holds true, since the Toeplitz operator $T_{\hat{A}}$ is given by \cref{equation: definition of toeplitz operator with matrix-valued symbol} with $F := \hat{A}.$
\end{proof}

\begin{proof}[Proof of \cref{theorem: uniform operators}]
Let $A$ be a uniform operator on $\cH = \ell^2(\Z, \C^n)$ of the form \cref{equation: characterisation of uniformity}, and let $\hat{A}$ be given by \cref{equation: definition of f}. It follows from \cref{equation: essential spectra of uniform operators} that $A_\rL \cong T_{\hat{A}}$ and $A_\rR \cong T_{\hat{A}(-^*)},$ where $\cong$ denotes unitary equivalence. The Fredholmness and essential spectra are invariant under unitary transforms. 

(i) It follows from \cref{corollary: properties of strictly local operators} (i) that the operator $A$ is Fredholm if and only if $A_\rL \cong T_{\hat{A}}, A_{\rR} \cong T_{\hat{A}(-^*)}$ are both Fredholm, and in this case we have $\ind A = \ind T_{\hat{A}} + \ind T_{\hat{A}(-^*)}.$ On the other hand, it follows from \cref{theorem: Gohberg-Krein}(i) that $A_\rL \cong T_{\hat{A}}$ is Fredholm if and only if $A_{\rR} \cong T_{\hat{A}(-^*)}$ is Fredholm if and only if $\det \hat{A}$ is nowhere vanishing. In this case, we have $\ind A_\rL = -\wn \left(\det \hat{A}\right)$ and $\ind A_\rR = -  \wn \left(\det \hat{A}(-^*)\right).$ Therefore,
\[
\ind A
= \ind A_\rL + \ind A_\rR
= -\wn \left(\det \hat{A}\right) -  \wn \left(\det \hat{A}(-^*)\right)
= -\wn \left(\det \hat{A}\right) + \wn \left(\det \hat{A} \right) 
= 0.
\]

(ii) It follows from \cref{corollary: properties of strictly local operators} (ii) that $\ess(A) = \ess(A_{\rL}) \cup \ess(A_{\rR}),$ where $A_\rL \cong T_{\hat{A}}, A_{\rR} \cong T_{\hat{A}(-^*)}.$ It follows from \cref{theorem: Gohberg-Krein}(ii) that
\[
\ess(A_\rL) = \ess(T_{\hat{A}}) = \bigcup_{z \in \T} \sigma \left(\hat{A}(z)\right) = \bigcup_{z \in \T} \sigma \left(\hat{A}(z^*) \right) = \ess\left(T_{\hat{A}(-^*)}\right) = \ess(A_\rR),
\]
where the third equality follows from the fact that the ranges of $\hat{A}, \hat{A}(-^*)$ are identical.
\end{proof}

\subsection{Proof of Theorem A}

We are now in a position to prove \cref{theorem: two-phase case} with the aid of the following lemma;

\begin{lemma}
\label{lemma: theorem A}
Under the assumption of \cref{theorem: two-phase case}, the difference $A_{\rL} \oplus A_{\rR} - A(-\infty)_{\rL} \oplus A(+\infty)_{\rR}$ is compact.
\end{lemma}
It follows that $A_{\rL} - A(-\infty)_{\rL}, A_{\rR} - A(+\infty)_{\rR}$ are both compact.
\begin{proof}
Let $P_{+\infty} := \bigoplus_{j=1}^n P_{\rR},$ and let $P_{-\infty} := \bigoplus_{j=1}^n P_{\rL}.$ We have
\begin{align*}
A_{\rL} \oplus A_{\rR} - A(-\infty)_{\rL} \oplus A(+\infty)_{\rR}
&= A_{\rL} \oplus 0 + 0 \oplus A_{\rR} - (A(-\infty)_{\rL} \oplus 0 + 0 \oplus A(+\infty)_{\rR}) \\
&= P_{-\infty}AP_{-\infty} +  P_{+\infty}AP_{+\infty} - (P_{-\infty} A(-\infty) P_{-\infty} +  P_{+\infty}A(+\infty)P_{+\infty}) \\
&=P_{-\infty}(A - A(-\infty))P_{-\infty} +  P_{+\infty}(A - A(+\infty))P_{+\infty}.
\end{align*}
It remains to show that $B(\star) := P_{\star}(A -  A(\star))P_{\star}$ is compact for each $\star = \pm \infty.$ We have
\[
B(\star) =
P_{\star}
\begin{pmatrix}
\sum^k_{y=-k} (a_{11}(y, \cdot) - a_{11}(y, \star)) L^{y} & \dots & \sum^k_{y=-k} (a_{1n}(y, \cdot) - a_{1n}(y, \star)) L^{y} \\
\vdots & \ddots& \vdots \\
\sum^k_{y=-k} (a_{n1}(y, \cdot) - a_{n1}(y, \star)) L^{y} & \dots & \sum^k_{y=-k} (a_{nn}(y, \cdot) - a_{nn}(y, \star)) L^{y} \\
\end{pmatrix}
P_{\star}.
\]
If $B(\star) = (B_{ij}(\star))$ is the block-operator matrix representation of $B,$ then
\begin{equation}
\label{equation: definition of Bij}
B_{ij}(\star) =
\begin{cases}
\sum^k_{y=-k} P_{\rR}(a_{ij}(y, \cdot) - a_{ij}(y, +\infty)) L^{y} P_{\rR}, & \star = +\infty, \\
\sum^k_{y=-k} P_{\rL}(a_{ij}(y, \cdot) - a_{ij}(y, -\infty)) L^{y} P_{\rL}, & \star = -\infty.
\end{cases}
\end{equation}
Note that the projection $P_{\rR}$ on $\ell^2(\Z) = \bigoplus_{x \in \Z} \C$ is a multiplication operator of the form $P_{\rR} = \bigoplus_{x \in \Z} \delta_{\rR}(x),$ where $\delta_{\rR}(x) := 1$ for each $x \geq 0$ and $\delta_{\rR}(x) := 0$ for each $x < 0.$ That is, for each $y \in \{-k, \dots, k\},$
\begin{align*}
&P_{\rR}(a_{ij}(y, \cdot) - a_{ij}(y, +\infty)) = \bigoplus_{x \in \Z} \delta_{\rR}(x)(a_{ij}(y, x) - a_{ij}(y, +\infty)), \\
&\lim_{x \to \pm \infty}  \delta_{\rR}(x)(a_{ij}(y, x) - a_{ij}(y, +\infty)) = 0.
\end{align*}
It follows that $P_{\rR}(a_{ij}(y, \cdot) - a_{ij}(y, +\infty))$ is compact for each $i,j,y,$ and so $B_{ij}(+\infty)$ given by \cref{equation: definition of Bij} is compact. Hence, $B(+\infty) = (B_{ij}(+\infty))$ is compact. An analogous argument can be used to show that $B(-\infty) = (B_{ij}(-\infty))$ is compact. The claim follows.
\end{proof}

\begin{proof}[Proof of \cref{theorem: two-phase case}]
Under the assumption of \cref{theorem: two-phase case}, it follows from \cref{corollary: properties of strictly local operators} and \cref{lemma: theorem A} that $A - A(-\infty)_{\rL} \oplus A(+\infty)_{\rR}$ is compact, where $A(-\infty), A(+\infty)$ are uniform operators.

(i) We shall make use of \cref{theorem: uniform operators} (i). We have that $A$ is Fredholm if and only if $A(-\infty)_{\rL}, A(+\infty)_{\rR}$ are Fredholm if and only if $\hat{A}(\cdot, \star)$ is nowhere vanishing on $\T$ for each $\star = \pm \infty.$ In this case, we have
\begin{align*}
\ind A
= \ind A(-\infty)_{\rL} + \ind A(+\infty)_{\rR} 
= -\wn\left(\hat{A}(\cdot,-\infty)\right) + \wn\left(\hat{A}(\cdot,+\infty) \right).
\end{align*}

(ii) We shall make use of \cref{theorem: uniform operators} (ii). We have
\[
\ess(A)
= \ess(A(-\infty)_{\rL}) \cup \ess(A(+\infty)_{\rR})
= \ess(A(-\infty)) \cup \ess(A(+\infty)),
\]
where each $\ess(A(\pm \infty))$ is given by the formula \cref{equation2: essential spectrum of two-phase operator}.
\end{proof}

\section{Applications of Theorem A}
\label{section: examples}

Quantum walk theory is a quantum-mechanical counterpart of the classical random walk theory, and it constitutes a vast multidisciplinary branch of modern Science. Certain primitive forms of this ubiquitous notion can be found, for example, in \cite{Gudder-1988,Aharonov-Davidovich-Zagury-1993,Meyer-1996,Ambainis-Bach-Nayak-Vishwanath-Watrous-2001}. Mathematically rigorous studies of discrete-time quantum walks include: scattering-theoretic analysis \cite{Suzuki-2016,Richard-Suzuki-Tiedra-2017,Richard-Suzuki-Tiedra-2018,Morioka-2019,Wada-2020}, non-linear analysis \cite{Maeda-Sasaki-Segawa-Suzuki-Suzuki-2018a,Maeda-Sasaki-Segawa-Suzuki-Suzuki-2018b,Maeda-Sasaki-Segawa-Suzuki-Suzuki-2019}, 
localisation and weak-limit theorems \cite{Konno-2002,Inui-Konishi-Konno-2004,Segawa-2011,Cantero-Grunbaum-Moral-Velazquez-2012,Suzuki-2016,Fuda-Funakawa-Suzuki-2018,Fuda-Funakawa-Suzuki-2019}, classification theorems \cite{Ohno-2016,Ohno-2017,Cedzich-Geib-Grunbaum-Stahl-Velazquez-Werner-Werner-2018}, discrete analogues of the time-operator \cite{Sambou-Tiedra-2019,Funakawa-Matsuzawa-Sasaki-Suzuki-Teranishi-2020}, and index theorems \cite{Cedzich-Geib-Stahl-Velazquez-Werner-Werner-2018,Suzuki-Tanaka-2019,Matsuzawa-2020}.

In this section we shall give a new index theorem in align with the setting of \cite{Cedzich-Geib-Grunbaum-Stahl-Velazquez-Werner-Werner-2018,Cedzich-Geib-Stahl-Velazquez-Werner-Werner-2018,Suzuki-Tanaka-2019,Matsuzawa-2020} as a direct application of \cref{theorem: two-phase case}(i). More precisely, we consider a concrete quantum walk model defined on the underlying Hilbert space $\ell^2(\Z, \C^2),$ which unifies all of the following models: the one-dimensional quantum walk considered in \cite{Ambainis-Bach-Nayak-Vishwanath-Watrous-2001,Konno-2002,Suzuki-2016}, Kitagawa's split-step quantum walk \cite{Kitagawa-Rudner-Berg-Demler-2010,Kitagawa-Broome-Fedrizzi-Rudner-Berg-Kassal-Aspuru-Demler-White-2012,Kitagawa_2012}, another split-step quantum walk in  \cite{Fuda-Funakawa-Suzuki-2017,Fuda-Funakawa-Suzuki-2018,Fuda-Funakawa-Suzuki-2019,Suzuki-Tanaka-2019,Matsuzawa-2020}. We compute the following two associated topological invariants: (i) a certain well-defined Fredholm index known as the \textbi{Witten index}, and (ii) the essential spectrum of the evolution operator. The complete classification of these invariants can be found in \cref{theorem: split-step}(i),(ii) respectively. The precise statement of \cref{theorem: split-step}, including the definition of the model, appears in \cref{section: statement of the main theorem}. Proofs of \cref{theorem: split-step}(i),(ii) will be given in \cref{section: fredholm index},\cref{section: essential spectrum} respectively. This paper concludes with several remarks in \cref{section: concluding remarks}.

\subsection{Statement of the main theorem (Theorem B)}
\label{section: statement of the main theorem}

What follows is a brief overview of \cite[\textsection 2]{Suzuki-2019} and \cite[\textsection 2.1]{Suzuki-Tanaka-2019}. An operator $\varGamma$ defined on an abstract Hilbert space $\cH$ is called an \textbi{involution}, if $\varGamma^2 = 1,$ where $1$ denotes the identity operator on $\cH.$ Note first that any unitary involution is self-adjoint (in fact, if an operator possess any two of the properties ``self-adjoint'', ``unitary'', ``involutory'', then it automatically has the third). The two operators $-1,1$ are referred to as \textbi{trivial unitary involutions}. We have $\cH = \ker(\varGamma - 1) \oplus \ker(\varGamma + 1)$ for any unitary involution $\varGamma.$ The model we shall consider in this section is based on the following simple finite-dimensional example;

\begin{eg}
\label{example: unitary involution}
If $X$ denotes a $2 \times 2$ matrix viewed as an operator on $\C^2,$ then $X$ is a non-trivial unitary involution if and only if there exist $\alpha \in \R$ and $\beta \in \C$ satisfying the following equalities:
\begin{equation}
\label{equation: unitary involutory matric}
X =
\begin{pmatrix}
\alpha & \beta \\
\beta^* & -\alpha
\end{pmatrix}, \qquad \alpha^2 + |\beta|^2 = 1.
\end{equation}
Note that the spectrum of such a matrix $X$ is $\{1,-1\},$ and so the following diagonalisation is possible;
\begin{equation}
\label{equation: diagonalisation of unitary involutory matrix}
\eta^*
\begin{pmatrix}
\alpha & \beta \\
\beta^* & -\alpha
\end{pmatrix}
\eta
=
\begin{pmatrix}
1 & 0 \\
0 & -1
\end{pmatrix}, \qquad
\eta :=
\frac{1}{\sqrt{2}}
\begin{pmatrix}
1 & 0 \\
0 & e^{-i\Theta}
\end{pmatrix}
\begin{pmatrix}
\alpha_+ & -\alpha_- \\
\alpha_-   & \alpha_+
\end{pmatrix},
\end{equation}
where $\Theta$ is any real number satisfying $\beta = e^{i \Theta} |\beta|$ and $\alpha_\pm := \sqrt{1 \pm \alpha}.$ Note that $\eta$ is unitary.
\end{eg}

A \textbi{chiral pair} on $\cH$ is any pair $(\varGamma, \varGamma')$ of two unitary involutions $\varGamma, \varGamma'$ on $\cH.$ Given such a pair $(\varGamma, \varGamma'),$ the operator $U := \varGamma \varGamma'$ is called the \textbi{evolution operator} of $(\varGamma, \varGamma').$ Let $R, Q$ be the real and imaginary parts of $U$ respectively:
\[
R := \Re U = \frac{U + U^*}{2} = \frac{\{ \varGamma, \varGamma' \}}{2}, \qquad
Q := \Im U = \frac{U - U^*}{2i} = \frac{[ \varGamma, \varGamma']}{2i},
\]
where $\{X,Y\} := XY + YX$ and $[X,Y] := XY - YX.$ Note that the evolution operator $U$ satisfies the following \textbi{chiral symmetry conditions} with respect to $\varGamma, \varGamma'$ respectively:
\begin{align}
\label{equation: chiral symmetry}
U^* = \varGamma' \varGamma = \varGamma^2 \varGamma' \varGamma = \varGamma (\varGamma \varGamma')\varGamma = \varGamma U \varGamma, \qquad
U^* = \varGamma' \varGamma = \varGamma' \varGamma (\varGamma')^2 = \varGamma' (\varGamma \varGamma') \varGamma' = \varGamma' U \varGamma'.
\end{align}
The above equality immediately implies the commutation relations $[\varGamma, R] = [\varGamma', R] = 0$ and anti-commutation relations $\{\varGamma, Q\} = \{\varGamma', Q\} = 0,$ and so $R, Q, U$ admit the following block-operator matrix representations with respect to $\cH = \ker(\varGamma - 1) \oplus \ker(\varGamma + 1)  = \ker(\varGamma' - 1) \oplus \ker(\varGamma' + 1):$
\begin{align}
\label{equation: representation of R}
R &=
\begin{pmatrix}
R_1 & 0 \\
0 & R_2
\end{pmatrix}_{\ker(\varGamma - 1) \oplus \ker(\varGamma + 1)}
=
\begin{pmatrix}
R'_1 & 0 \\
0 & R'_2
\end{pmatrix}_{\ker(\varGamma' - 1) \oplus \ker(\varGamma' + 1)}, \\
\label{equation: representation of Q}
Q &=
\begin{pmatrix}
0 & Q_0^* \\
Q_0 & 0
\end{pmatrix}_{\ker(\varGamma - 1) \oplus \ker(\varGamma + 1)}
=
\begin{pmatrix}
0 & (Q'_0)^* \\
Q'_0 & 0
\end{pmatrix}_{\ker(\varGamma' - 1) \oplus \ker(\varGamma' + 1)}, \\
\label{equation: representation of U}
U &=
\begin{pmatrix}
R_1 & iQ_0^* \\
iQ_0 & R_2
\end{pmatrix}_{\ker(\varGamma - 1) \oplus \ker(\varGamma + 1)}
=
\begin{pmatrix}
R'_1 & i(Q'_0)^* \\
iQ'_0 & R'_2
\end{pmatrix}_{\ker(\varGamma' - 1) \oplus \ker(\varGamma' + 1)},
\end{align}
where $R_j, R'_j$ are self-adjoint for each $j=1,2.$ The chiral pair $(\varGamma, \varGamma')$ is said to be \textbi{Fredholm}, if the operator $Q$ is Fredholm (or, equivalently, $Q_0, Q'_0$ are Fredholm). In this case, we define the \textbi{Witten indices} of the pairs $(\varGamma, \varGamma')', (\varGamma, \varGamma)$ by the following formulas respectively:
\[
\ind(\varGamma, \varGamma') := \ind Q_0, \qquad \ind(\varGamma', \varGamma) := \ind Q'_0,
\]
where $\ind Q_0, \ind Q'_0$ denote the Fredholm indices of $Q_0, Q'_0$ respectively. The Witten indices introduced above are \textbi{unitarily invariant} in the following precise sense (see \cite[Theorem 3]{Suzuki-Tanaka-2019} for details); for any unitary operator $\epsilon$ on $\cH,$ we have that $(\varGamma, \varGamma')$ is Fredholm if and only if so is $(\epsilon^* \varGamma \epsilon, \epsilon^* \varGamma' \epsilon),$ and that in this case
\begin{equation}
\label{equation: unitary invariance of the witten index}
\ind(\varGamma, \varGamma') = \ind(\epsilon^* \varGamma \epsilon, \epsilon^* \varGamma' \epsilon), \qquad
\ind(\varGamma', \varGamma) = \ind(\epsilon^* \varGamma' \epsilon, \epsilon^* \varGamma \epsilon).
\end{equation}

With \cref{example: unitary involution} in mind, we are now in a position to state the main theorem of the current section;

\begin{mtheorem}
\label{theorem: split-step}
Let $\varGamma,\varGamma'$ be two unitary involutions defined respectively as the following block-operator matrices with respect to $\ell^2(\Z, \C^2) = \ell^2(\Z) \oplus \ell^2(\Z):$
\begin{align}
\label{equation: definition of shift operator}
\tag{B1}
\varGamma &:=
\begin{pmatrix}
1 & 0\\
0  &  L^*
\end{pmatrix}
\begin{pmatrix}
p   & q \\
q^* & -p
\end{pmatrix}
\begin{pmatrix}
1 & 0\\
0  &  L
\end{pmatrix}
=
\begin{pmatrix}
p & q L \\
 L^*q^*  & -p(\cdot - 1)
\end{pmatrix}, \\
\label{equation: definition of coin operator}
\tag{B2}
\varGamma' &:=
\begin{pmatrix}
a & b^* \\
b & -a
\end{pmatrix},
\end{align}
where we assume that two convergent $\R$-valued sequences $p = (p(x))_{x \in \Z}, a = (a(x))_{x \in \Z}$ and two convergent $\C$-valued sequences $q = (q(x))_{x \in \Z}, b = (b(x))_{x \in \Z}$ satisfy the following:
\begin{align}
\label{equation: p and q}
\tag{B3}
&p(x)^2 + |q(x)|^2 = 1, & &x \in \Z, \\
\label{equation: alpha and beta}
\tag{B4}
&a(x)^2 + |b(x)|^2 = 1, & &x \in \Z,  \\
\label{equation: limits of real sequences}
\tag{B5}
&p(\pm \infty) := \lim_{x \to \pm \infty} p(x) \in \R, && a(\pm \infty) := \lim_{x \to \pm \infty} a(x) \in \R, \\
\label{equation: limits of complex sequences}
\tag{B6}
&q(\pm \infty) := \lim_{x \to \pm \infty} q(x) \in \C,& &b(\pm \infty) := \lim_{x \to \pm \infty} b(x) \in \C, \\
\label{equation: limits of theta}
\tag{B7}
&\theta(\pm \infty) :=
\begin{cases}
\arg q(\pm \infty), & q(\pm \infty) \neq 0, \\
0,             & q(\pm \infty) = 0,
\end{cases} &
&
\phi(\pm \infty) :=
\begin{cases}
\arg b(\pm \infty), & b(\pm \infty) \neq 0, \\
0,                 & b(\pm \infty) = 0,
\end{cases}
\end{align}
where $\arg w$ of a non-zero complex number $w$ is uniquely defined by $w = e^{i \arg w}$ and $\arg w \in [0,2\pi).$ Then the following two assertions hold true:
\begin{enumerate}[(i)]
\item The chiral pair $(\varGamma,\varGamma' )$ is Fredholm if and only if $|p(\star)| \neq |a(\star)|$ for each $\star = \pm \infty.$ In this case, we have
\begin{align}
\label{equation1: Witten index formula}
\tag{B8}
&\ind (\varGamma,\varGamma') =
\begin{cases}
0, &  |p(-\infty)| < |a(-\infty)| \mbox{ and }  |p(+\infty)| < |a(+\infty)|, \\
+\sgn p(+\infty), &  |p(-\infty)| < |a(-\infty)| \mbox{ and }  |p(+\infty)| > |a(+\infty)|, \\
-  \sgn p(-\infty), &  |p(-\infty)| > |a(-\infty)| \mbox{ and } |p(+\infty)| < |a(+\infty)|, \\
+\sgn p(+\infty) - \sgn p(-\infty), &  |p(-\infty)| >  |a(-\infty)| \mbox{ and }  |p(+\infty)| > |a(+\infty)|,
\end{cases} \\
\label{equation2: Witten index formula}
\tag{B9}
&\ind (\varGamma',\varGamma) =
\begin{cases}
-\sgn a(+\infty) + \sgn a(-\infty), &  |p(-\infty)| < |a(-\infty)| \mbox{ and }  |p(+\infty)| < |a(+\infty)|, \\
+ \sgn a(-\infty), &  |p(-\infty)| < |a(-\infty)| \mbox{ and }  |p(+\infty)| > |a(+\infty)|, \\
- \sgn a(+\infty), &  |p(-\infty)| > |a(-\infty)| \mbox{ and }  |p(+\infty)| < |a(+\infty)|, \\
0, &  |p(-\infty)| >  |a(-\infty)| \mbox{ and }  |p(+\infty)| > |a(+\infty)|,
\end{cases}
\end{align}
where the sign function $\sgn : \R \to \{-1,1\}$ is defined by
\begin{equation}
\tag{B10}
\label{equation: definition of sign function}
\sgn x :=
\begin{cases}
\frac{x}{|x|}, & x \neq 0, \\
1, & x = 0.
\end{cases}
\end{equation}

\item The essential spectrum of the evolution operator $U := \varGamma \varGamma'$ is given by
\begin{align}
\label{equation: essential spectrum of U}
\tag{B11}
\ess(U) &= \bigcup_{\star = \pm \infty} \left\{z \in \T \mid \sgn(p(\star) a(\star)) \cdot \Re z \in I(\star) \right\},  \\
\label{equation: definition of Istar}
\tag{B12}
I(\star) &:= [|p(\star) a(\star)| - |q(\star) b(\star)|, |p(\star) a(\star)| + |q(\star) b(\star)|], \qquad \star = \pm \infty.
\end{align}
Moreover, $\ess(U)$ does not contain both $-1,+1$ if and only if $|p(\star)| \neq |a(\star)|$ for each $\star = \pm \infty.$
\end{enumerate}
\end{mtheorem}

The chiral pair $(\varGamma,\varGamma')$ in \cref{theorem: split-step} is a one-dimensional split-step quantum walk model considered in \cite{Fuda-Funakawa-Suzuki-2017,Fuda-Funakawa-Suzuki-2018,Fuda-Funakawa-Suzuki-2019,Suzuki-Tanaka-2019,Matsuzawa-2020} with a modification that the two parameters $p, q$ depend freely on $\Z.$ Note that this seemingly minor modification leads to the new Witten index formula \cref{equation1: Witten index formula} taking values from $\{-2,-1,0,+1,+2\},$ where the indices $\pm 2$ never appear in the existing formula (see \cite[(A2)]{Suzuki-Tanaka-2019} or \cite[Theorem 1.1]{Matsuzawa-2020} for details), since $p, q$ are kept constant in the five papers mentioned above. The index formula \cref{equation2: Witten index formula} is new to the best of the author's knowledge. As we shall see in \cref{section: topologically protected bound states}, the two index formulas \crefrange{equation1: Witten index formula}{equation2: Witten index formula} naturally provide a lower bound for $\dim \ker(U \mp 1).$

\subsection{Proof of Theorem B (i)}
\label{section: fredholm index}

What follows is a generalisation of \cite[\textsection 3]{Suzuki-Tanaka-2019}. With the notation introduced in \cref{theorem: split-step}, on one hand, the imaginary part $Q$ admits two off-diagonal matrix representations as in \cref{equation: representation of Q}, where $\ind(\varGamma,\varGamma') = \ind Q_0$ and $\ind(\varGamma',\varGamma) = \ind Q'_0$ by definition. On the other hand, an easy computation shows that the same operator $Q$ does \textit{not} admit an off-diagonal representation with respect to the orthogonal decomposition $\ell^2(\Z,\C^2) = \ell^2(\Z) \oplus \ell^2(\Z).$ The unitary invariance property \cref{equation: unitary invariance of the witten index} motivates us to construct explicit unitary operators $\epsilon, \gamma : \ell^2(\Z) \to \ell^2(\Z),$ such that $\epsilon^* Q \epsilon, \gamma^* Q \gamma$ become off-diagonal with respect to this decomposition. 

\begin{lemma}
\label{lemma: wada decomposition}
Let $(\varGamma,\varGamma')$ be the chiral pair in \cref{theorem: split-step}, and let $U := \varGamma \varGamma'$ be the associated evolution operator. Let $R, Q$ be the real and imaginary parts of $U$ respectively. For each $x \in \Z,$ let $\theta(x), \phi(x)$ be any real numbers satisfying $q(x) = |q(x)|e^{i \theta(x)}$ and $b(x) = |b(x)|e^{i \phi(x)}.$ Let $p_\pm := \sqrt{1 \pm p},$ and let $a_\pm := \sqrt{1 \pm a}.$ Let
\begin{equation}
\label{equation: wada transform}
\epsilon :=
\frac{1}{\sqrt{2}}
\begin{pmatrix}
1 & 0 \\
0  &  L^*e^{-i \theta}
\end{pmatrix}
\begin{pmatrix}
p_+  & -p_- \\
p_-  & p_+
\end{pmatrix}, \qquad
\gamma :=
\frac{1}{\sqrt{2}}
\begin{pmatrix}
1 & 0 \\
0 & e^{i \phi}
\end{pmatrix}
\begin{pmatrix}
a_+ & -a_- \\
a_-   & a_+
\end{pmatrix}.
\end{equation}
Then the unitary operators $\epsilon, \gamma$ give the following decompositions with respect to $\ell^2(\Z,\C^2) = \ell^2(\Z) \oplus \ell^2(\Z):$
\begin{align}
\label{equation1: wada decomposition}
&\epsilon^* \varGamma \epsilon
=
\begin{pmatrix}
1 & 0 \\
0 & -1
\end{pmatrix},
&&\epsilon^* U \epsilon
=
\begin{pmatrix}
R_{\epsilon_1} & iQ_{\epsilon_0}^* \\
iQ_{\epsilon_0} & R_{\epsilon_2}
\end{pmatrix},
&&\epsilon^* R \epsilon
=
\begin{pmatrix}
R_{\epsilon_1} & 0 \\
0 & R_{\epsilon_2}
\end{pmatrix},
&&\epsilon^* Q \epsilon
=
\begin{pmatrix}
0 & Q_{\epsilon_0}^* \\
Q_{\epsilon_0} & 0
\end{pmatrix}, \\
\label{equation2: wada decomposition}
&\gamma^* \varGamma' \gamma
=
\begin{pmatrix}
1 & 0 \\
0 & -1
\end{pmatrix},
&&\gamma^* U \gamma
=
\begin{pmatrix}
R_{\gamma_1} & iQ_{\gamma_0}^* \\
iQ_{\gamma_0} & R_{\gamma_2}
\end{pmatrix},
&&\gamma^* R \gamma
=
\begin{pmatrix}
R_{\gamma_1} & 0 \\
0 & R_{\gamma_2}
\end{pmatrix},
&&\gamma^* Q \gamma
=
\begin{pmatrix}
0 & Q_{\gamma_0}^* \\
Q_{\gamma_0} & 0
\end{pmatrix},
\end{align}
where the six operators $Q_{\epsilon_0},Q_{\gamma_0},R_{\epsilon_1},R_{\gamma_1},R_{\epsilon_2},R_{\gamma_2}$ are defined respectively by the following formulas:
\begin{align}
\label{equation: definition of Qepsilon}
-2i Q_{\epsilon_0} &:= p_+  e^{i \theta}L b p_+ - p_- b^*  L^* e^{-i \theta}p_-  - |q|(a + a(\cdot + 1)), \\
\label{equation: definition of Qgamma}
2i Q_{\gamma_0} &:=  a_+ e^{-i\phi}L^* q^* a_+ - a_- q L e^{i\phi} a_- - |b|(p + p(\cdot - 1)), \\
\label{equation1: definition of Repsilon}
2R_{\epsilon_1} &:= p_-  e^{i \theta}L b p_+ + p_+ b^*  L^* e^{-i \theta} p_-  + p_+^2 a - p_-^2 a(\cdot + 1), \\
\label{equation1: definition of Rgamma}
2R_{\gamma_1} &:= a_- e^{-i\phi}L^* q^* a_+ + a_+ q L e^{i\phi} a_- + a_+^2p - a_-^2p(\cdot - 1), \\
\label{equation2: definition of Repsilon}
2R_{\epsilon_2} &:= p_+  e^{i \theta}L b p_- + p_- b^*  L^* e^{-i \theta} p_+  - p_-^2a + p_+^2 a(\cdot + 1), \\
\label{equation2: definition of Rgamma}
2R_{\gamma_2} &:= a_+ e^{-i\phi}L^* q^* a_- + a_- q L e^{i\phi} a_+ - a_-^2p + a_+^2p(\cdot - 1).
\end{align}
Moreover, the chiral pair $(\varGamma,\varGamma')$ is Fredholm if and only if $Q_{\epsilon_0},Q_{\gamma_0}$ are Fredholm. In this case, 
\begin{equation}
\label{equation: first index formula}
\ind(\varGamma,\varGamma') = \ind Q_{\epsilon_0}, \qquad \ind(\varGamma',\varGamma) = \ind Q_{\gamma_0}.
\end{equation}
\end{lemma}
\begin{proof}
It follows from \cref{equation: diagonalisation of unitary involutory matrix} that we have the following diagonalisation:
\begin{align}
\epsilon_0^*
\begin{pmatrix}
p & q \\
q^* & -p
\end{pmatrix}
\epsilon_0
&=
\begin{pmatrix}
1 & 0 \\
0 & -1
\end{pmatrix},
&
\epsilon_0 &:=
\frac{1}{\sqrt{2}}
\begin{pmatrix}
1 & 0 \\
0 & e^{-i\theta}
\end{pmatrix}
\begin{pmatrix}
p_+ & -p_- \\
p_-   & p_+
\end{pmatrix}, \\
\gamma^*
\begin{pmatrix}
a & b^* \\
b & -a
\end{pmatrix}
\gamma
&=
\begin{pmatrix}
1 & 0 \\
0 & -1
\end{pmatrix},
&
\gamma &:=
\frac{1}{\sqrt{2}}
\begin{pmatrix}
1 & 0 \\
0 & e^{i \phi}
\end{pmatrix}
\begin{pmatrix}
a_+ & -a_- \\
a_-   & a_+
\end{pmatrix}.
\end{align}
The operator $\epsilon$ given by the first equality in \cref{equation: wada transform} can be written as the product $\epsilon = (1 \oplus L^*) \epsilon_0$ of two unitary operators $1 \oplus L^*$ and $\epsilon_0.$ With the first equality in \cref{equation: definition of shift operator}, we obtain
\[
\epsilon^* \varGamma \epsilon
=
\epsilon_0^*
\begin{pmatrix}
1 & 0 \\
0  &  L
\end{pmatrix}
\left(
\begin{pmatrix}
1 & 0 \\
0  &  L^*
\end{pmatrix}
\begin{pmatrix}
p   & q \\
q^* & -p
\end{pmatrix}
\begin{pmatrix}
1 & 0 \\
0  &  L
\end{pmatrix}
\right)
\begin{pmatrix}
1 & 0 \\
0  &  L^*
\end{pmatrix}
\epsilon_0 \\
=
\begin{pmatrix}
1 & 0 \\
0 & -1
\end{pmatrix}.
\]
Given an operator $X$ on $\ell^2(\Z,\C^2),$ we introduce the shorthand $X_\epsilon := \epsilon^* X \epsilon$ and $X_\eta := \eta^* X \eta.$ With this convention in mind, we have the commutation relations $[\varGamma_\epsilon, R_\epsilon] = [\varGamma'_\gamma, R_\gamma] = 0$ and anti-commutation relations $\{\varGamma_\epsilon, Q_\epsilon\} = \{\varGamma'_\gamma, Q_\gamma\} = 0,$ where $\varGamma_\epsilon = \varGamma'_\gamma = 1 \oplus -1$ with respect to $\ell^2(\Z,\C^2) = \ell^2(\Z) \oplus \ell^2(\Z).$ It follows that we have the following representations:
\begin{align}
\label{equation: epsilon representation}
R_\epsilon
&=
\begin{pmatrix}
R'_{\epsilon_1} & 0 \\
0 & R'_{\epsilon_2}
\end{pmatrix},
&
Q_\epsilon
&=
\begin{pmatrix}
0 & (Q'_{\epsilon_0})^* \\
Q'_{\epsilon_0} & 0
\end{pmatrix},
&
U_\epsilon
&= R_\epsilon + iQ_\epsilon =
\begin{pmatrix}
R'_{\epsilon_1} & i(Q'_{\epsilon_0})^* \\
iQ'_{\epsilon_0} & R'_{\epsilon_2}
\end{pmatrix}, \\
\label{equation: gamma representation}
R_\gamma
&=
\begin{pmatrix}
R'_{\gamma_1} & 0 \\
0 & R'_{\gamma_2}
\end{pmatrix},
&
Q_\gamma
&=
\begin{pmatrix}
0 & (Q'_{\gamma_0})^* \\
Q'_{\gamma_0} & 0
\end{pmatrix},
&
U_\gamma
&= R_\gamma + iQ_\gamma =
\begin{pmatrix}
R'_{\gamma_1} & i(Q'_{\gamma_0})^* \\
iQ'_{\gamma_0} & R'_{\gamma_2}
\end{pmatrix}.
\end{align}
It remains to show that the six operators introduced above coincide with the ones defined by the formulas \crefrange{equation: definition of Qepsilon}{equation2: definition of Rgamma}. Note that
\begin{align}
\label{equation1: Cepsilon}
2 \varGamma'_\epsilon
&=
2\varGamma_\epsilon U_\epsilon
=
2
\begin{pmatrix}
1 & 0 \\
0 & -1
\end{pmatrix}
\begin{pmatrix}
R'_{\epsilon_1} & i(Q'_{\epsilon_0})^* \\
iQ'_{\epsilon_0}   & R'_{\epsilon_2}  \\
\end{pmatrix}
=
\begin{pmatrix}
2R'_{\epsilon_1}     & 2i(Q'_{\epsilon_0})^* \\
-2iQ'_{\epsilon_0}   & -2R'_{\epsilon_2}  \\
\end{pmatrix},\\
\label{equation1: Sgamma}
2\varGamma_\gamma
&=
2 U_\gamma \varGamma'_\gamma
=
2
\begin{pmatrix}
R'_{\gamma_1}    & i(Q'_{\gamma_0})^* \\
iQ'_{\gamma_0}   & R'_{\gamma_2}  \\
\end{pmatrix}
\begin{pmatrix}
1 & 0 \\
0 & -1
\end{pmatrix}
=
\begin{pmatrix}
2R'_{\gamma_1}    & -2i(Q'_{\gamma_0})^* \\
2iQ'_{\gamma_0}   & -2R'_{\gamma_2}  \\
\end{pmatrix}.
\end{align}
It remains to compute $2\varGamma'_\epsilon, 2\varGamma_\gamma$;
\begin{align*}
2\epsilon^*
\begin{pmatrix}
0 & b^* \\
b & 0 \\
\end{pmatrix}
\epsilon
&=
\begin{pmatrix}
p_-  e^{i \theta}L b p_+ + p_+ b^*  L^*e^{-i \theta} p_-    & -p_-  e^{i \theta}L b p_- + p_+ b^*  L^* e^{-i \theta} p_+   \\
p_+  e^{i \theta}L b p_+ - p_- b^*  L^*e^{-i \theta} p_-  &
-p_+  e^{i \theta}L b p_- - p_- b^*  L^*e^{-i \theta} p_+
\end{pmatrix}, \\
2\epsilon^*
\begin{pmatrix}
a & 0 \\
0 & -a \\
\end{pmatrix}
\epsilon
&=
\begin{pmatrix}
p_+^2 a - p_-^2 a(\cdot + 1) & -|q|(a + a(\cdot + 1)) \\
-|q|(a + a(\cdot + 1))  & p_-^2 a - p_+^2a(\cdot + 1)
\end{pmatrix}, \\
2\gamma^*
\begin{pmatrix}
0 & q L \\
L^*q^*  & 0
\end{pmatrix}
\gamma
&=
\begin{pmatrix}
a_- e^{-i\phi}L^* q^* a_+ + a_+ q L e^{i\phi} a_- & -a_- e^{-i\phi}L^* q^* a_- + a_+ q L e^{i\phi} a_+ \\
a_+ e^{-i\phi}L^* q^* a_+ - a_- q L e^{i\phi} a_-  & -a_+ e^{-i\phi}L^* q^* a_- - a_- q L e^{i\phi} a_+
\end{pmatrix}, \\
2\gamma^*
\begin{pmatrix}
p & 0 \\
0  & -p(\cdot - 1)
\end{pmatrix}
\gamma
&=
\begin{pmatrix}
a_+^2p - a_-^2p(\cdot - 1) & -|b|(p + p(\cdot - 1)) \\
-|b|(p + p(\cdot - 1))  & a_-^2p - a_+^2p(\cdot - 1)
\end{pmatrix}.
\end{align*}
It follows from the above equalities that
\begin{align}
\label{equation2: Cepsilon}
2\varGamma'_\epsilon &=
2\epsilon^*
\begin{pmatrix}
0 & b^* \\
b & 0 \\
\end{pmatrix}
\epsilon
+
2\epsilon^*
\begin{pmatrix}
a & 0 \\
0 & -a \\
\end{pmatrix}
=
\begin{pmatrix}
2R_{\epsilon_1} & 2i Q_{\epsilon_0}^* \\
-2i Q_{\epsilon_0} & -2R_{\epsilon_2}
\end{pmatrix}, \\
\label{equation2: Sgamma}
2\varGamma_\gamma &=
2\gamma^*
\begin{pmatrix}
0 & q L \\
L^*q^*  & 0
\end{pmatrix}
\gamma
+
2\gamma^*
\begin{pmatrix}
p & 0 \\
0  & -p(\cdot - 1)
\end{pmatrix}
\gamma
=
\begin{pmatrix}
2R_{\gamma_1} & -2i Q_{\gamma_0}^* \\
2i Q_{\gamma_0} & -2R_{\gamma_2}
\end{pmatrix}.
\end{align}
By comparing \crefrange{equation1: Cepsilon}{equation1: Sgamma} with \crefrange{equation2: Cepsilon}{equation2: Sgamma}, we see that \crefrange{equation1: wada decomposition}{equation2: wada decomposition} hold true.

Note that $\ell^2(\Z,\C^2) = \ell^2(\Z) \oplus \ell^2(\Z)$ can be identified with the orthogonal sum $\ell^2(\Z) \oplus \{0\} \oplus \{0\} \oplus \ell^2(\Z)$ through the following unitary transform;
\[
\ell^2(\Z,\C^2) \ni (\Psi_1, \Psi_2) \longmapsto (\Psi_1, 0,0,\Psi_2) \in \ell^2(\Z) \oplus \{0\} \oplus \{0\} \oplus \ell^2(\Z).
\]
It is then easy to see that the operator $Q_\epsilon$ admits the following block-operator matrix representations:
\begin{equation}
\label{equation1: representation of Qepsilon}
Q_\epsilon
=
\begin{pmatrix}
0 & Q_{\epsilon_0}^* \\
Q_{\epsilon_0} & 0
\end{pmatrix}_{\ell^2(\Z) \oplus \ell^2(\Z)}
=
\begin{pmatrix}
0 & 0 & 0 & Q_{\epsilon_0}^* \\
0 & 0 & \textbf{0} & 0 \\
0 & \textbf{0} & 0 & 0 \\
Q_{\epsilon_0} & 0 & 0 & 0
\end{pmatrix}_{\ell^2(\Z) \oplus \{0\} \oplus \{0\} \oplus \ell^2(\Z)},
\end{equation}
where $\textbf{0}$ denotes the zero operator of the form $\textbf{0} : \{0\} \to \{0\},$ and where $\ell^2(\Z) \oplus \{0\} = \ker(\varGamma_\epsilon - 1)$ and $\{0\} \oplus \ell^2(\Z) = \ker(\varGamma_\epsilon + 1).$ On the other hand, the same operator $Q_\epsilon$ is the imaginary part of $U_\epsilon = \varGamma_\epsilon \varGamma'_\epsilon,$ and so it admits the following off-diagonal block-operator matrix representation according to \cref{equation: representation of Q};
\begin{equation}
\label{equation2: representation of Qepsilon}
Q_\epsilon =
\begin{pmatrix}
0 & (Q''_{0})^* \\
Q''_{0} & 0
\end{pmatrix}_{\ker(\varGamma_\epsilon - 1) \oplus \ker(\varGamma_\epsilon + 1)}
=
\begin{pmatrix}
0 & (Q''_{0})^* \\
Q''_{0} & 0
\end{pmatrix}_{(\ell^2(\Z) \oplus \{0\}) \oplus (\{0\} \oplus \ell^2(\Z))}.
\end{equation}
It follows from \crefrange{equation1: representation of Qepsilon}{equation2: representation of Qepsilon} that $Q''_{0}$ is an off-diagonal block-operator matrix of the form;
\[
Q''_{0} =
\begin{pmatrix}
0  & \textbf{0} \\
Q_{\epsilon_0} & 0
\end{pmatrix}.
\]
Since $\textbf{0} : \{0\} \to \{0\}$ is a Fredholm operator of zero index, we have that $Q''_{0}$ is Fredholm if and only if $Q_{\epsilon_0}$ is Fredholm. In this case, we have $\ind Q''_{0} = \ind Q_{\epsilon_0} + \ind(\textbf{0}) = \ind Q_{\epsilon_0} + 0 = \ind Q_{\epsilon_0}.$ The first equality in \cref{equation: first index formula} follows from the unitary invariance of the Witten index \cref{equation: unitary invariance of the witten index}. An analogous argument can be used to show that the second equality in \cref{equation: first index formula} also holds true.
\end{proof}
\begin{remark}
The above derivation of \cref{equation: first index formula} only requires the sequences $p,q,a,b$ to be bounded, and so the existence of the two-sided limits \crefrange{equation: limits of real sequences}{equation: limits of complex sequences} turn out to be redundant. Note, however, that from here on we shall impose \crefrange{equation: limits of real sequences}{equation: limits of complex sequences} to prove the index formulas \crefrange{equation1: Witten index formula}{equation2: Witten index formula}.
\end{remark}

With \cref{lemma: wada decomposition} in mind, it remains to compute the Fredholm indices of the following operators:
\begin{align}
\label{equation2: definition of Qepsilon}
-2i Q_{\epsilon_0} &= p_+  p_+(\cdot + 1) b(\cdot + 1) e^{i \theta}L  - p_-p_-(\cdot - 1) b^*  e^{-i \theta(\cdot - 1)}L^*   - |q|(a + a(\cdot + 1)), \\
\label{equation2: definition of Qgamma}
2i Q_{\gamma_0}   &= a_+ a_+(\cdot - 1)q(\cdot - 1)^* e^{-i\phi}L^* - a_- a_-(\cdot + 1) q e^{i\phi(\cdot + 1)}  L     - |b|(p + p(\cdot - 1)),
\end{align}
where $\theta, \phi$ can be any $\R$-valued sequences satisfying $q(x) = |q(x)|e^{i \theta(x)}$ and $b(x) = |b(x)|e^{i \phi(x)}$ for each $x \in \Z.$ Note that \cref{theorem: two-phase case}(i) is not immediately applicable to the above strictly local operators, since it is not necessarily true that $\theta$ and $\phi$ are convergent. More precisely, for each $\star = \pm \infty,$ if $q(\star) \neq 0,$ then we can explicitly construct $\theta$ in such a way that $\theta(\star) = \lim_{x \to \star} \theta(x)$ holds true. On the other hand, if $q(\star) =  0,$ then the same conclusion cannot be drawn in general, because there are some pathological examples. Note that the same remark applies to $\phi.$ The purpose of the current section is to overcome this hindrance, which does not appear under the setting of \cite{Suzuki-Tanaka-2019,Matsuzawa-2020}, where $p,q$ are held constant.

\begin{lemma}
\label{lemma: phase problem}
The following assertions hold true:
\begin{enumerate}[(i)]
\item There exist two $\R$-valued sequences $\theta_+ = (\theta_+(x))_{x \in \Z}, \theta_- = (\theta_-(x))_{x \in \Z},$ such that
\begin{equation}
\label{equation: modified Qepsilon}
\begin{aligned}
e^{-i \theta_+} (-2i Q_{\epsilon_0}) e^{i \theta_-}
&= p_+ p_+(\cdot + 1) b(\cdot + 1) e^{i(\theta - \theta_+ + \theta_-(\cdot + 1))} L  \\
&- p_- p_-(\cdot - 1) b^* e^{-i (\theta(\cdot - 1) - \theta_-(\cdot - 1) + \theta_+ )} L^* \\
&- |q|(a + a(\cdot + 1))e^{i(\theta_- - \theta_+)},
\end{aligned}
\end{equation}
where the three coefficients of the above strictly local operator have the following limits for each $\star = \pm \infty:$
\begin{align}
\label{equation1: new phase}
&\lim_{x \to \star} \left( p_+(x) p_+(x + 1) b(x + 1)  e^{i (\theta(x) - \theta_+(x) + \theta_-(x + 1))} \right) =
(p(\star) + 1) b(\star)e^{i \theta(\star)}, \\
\label{equation2: new phase}
&\lim_{x \to \star} \left(- p_-(x) p_-(x - 1) b(x)^* e^{-i (\theta(x - 1) - \theta_-(x - 1)  + \theta_+(x))}\right) =
(p(\star) - 1) b(\star)^* e^{-i \theta(\star)}, \\
\label{equation3: new phase}
&\lim_{x \to \star}
\left(-|q(x)|(a(x) + a(x + 1))e^{i (\theta_-(x) - \theta_+(x))}\right)
= -2|q(\star)|a(\star).
\end{align}

\item There exist two $\R$-valued sequences $\phi_+ = (\phi_+(x))_{x \in \Z}, \phi_- = (\phi_-(x))_{x \in \Z},$ such that
\begin{equation}
\label{equation: modified Qgamma}
\begin{aligned}
e^{i \phi_+} (2i Q_{\gamma_0}) e^{-i \phi_-}
&= a_+ a_+(\cdot - 1) q(\cdot - 1)^* e^{-i(\phi - \phi_+ + \phi_-(\cdot - 1))} L^*  \\
&- a_- a_-(\cdot + 1) q e^{i (\phi(\cdot + 1) - \phi_-(\cdot + 1) + \phi_+ )} L \\
&- |b|(p + p(\cdot - 1))e^{i(\phi_+ - \phi_-)},
\end{aligned}
\end{equation}
where the three coefficients of the above strictly local operator have the following limits for each $\star = \pm \infty:$
\begin{align}
\label{equation4: new phase}
&\lim_{x \to \star} \left( a_+(x) a_+(x - 1) q(x - 1)^* e^{-i(\phi(x) - \phi_+(x) + \phi_-(x - 1))} \right) =
(a(\star) + 1) q(\star)^*e^{-i \phi(\star)}, \\
\label{equation5: new phase}
&\lim_{x \to \star} \left(- a_-(x) a_-(x + 1) q(x) e^{i (\phi(x + 1) - \phi_-(x + 1) + \phi_+(x) )}\right) =
(a(\star) - 1) q(\star) e^{i \phi(\star)}, \\
\label{equation6: new phase}
&\lim_{x \to \star}
\left(- |b(x)|(p(x) + p(x - 1))e^{i(\phi_+(x) - \phi_-(x))}\right)
= -2|b(\star)|p(\star).
\end{align}
\end{enumerate}
\end{lemma}
\begin{proof}
For each $x \in \Z$ we let
\[
\star(x) :=
\begin{cases}
+\infty, & x \geq 0, \\
-\infty, & x < 0,
\end{cases} \qquad
\theta_{\pm}(x) :=
\begin{cases}
\theta(x), & p(\star(x)) =    \pm 1, \\
0,         & p(\star(x)) \neq \pm1, \\
\end{cases} \qquad
\phi_{\pm}(x) :=
\begin{cases}
\phi(x),   & a(\star(x)) =    \pm 1, \\
0,         & a(\star(x)) \neq \pm 1.
\end{cases} 
\]
Note that \cref{equation: modified Qepsilon} follows from \cref{equation2: definition of Qepsilon}, and \cref{equation: modified Qgamma} follows from \cref{equation2: definition of Qgamma}. For each $x \in \Z$ we let
\[
\Lambda_1(x) := \theta(x) - \theta_+(x) + \theta_-(x + 1), \quad
\Lambda_2(x) := \theta(x - 1) - \theta_-(x - 1)  + \theta_+(x), \quad
\Lambda_3(x) := \theta_-(x)  - \theta_+(x).
\]

(i)  It suffices to prove the following equalities:
\begin{align}
\label{equation7: new phase}
&\lim_{x \to \star} \left(p_+(x) p_+(x + 1) e^{i \Lambda_1(x)} \right) =
(p(\star) + 1) e^{i \theta(\star)}, \\
\label{equation8: new phase}
&\lim_{x \to \star} \left(p_-(x) p_-(x - 1) e^{-i\Lambda_2(x)}\right) =
-(p(\star) - 1) e^{-i \theta(\star)}, \\
\label{equation9: new phase}
&\lim_{x \to \star}
\left(|q(x)|e^{i \Lambda_3(x)}\right) = |q(\star)|.
\end{align}
Let $\star = \pm \infty$ and $|x| > 1$ be fixed. If $|p(\star)| < 1,$ then $\theta_+(x) = \theta_-(x) = 0.$ In this case, \crefrange{equation4: new phase}{equation6: new phase} follow from the fact that as $x \to \star$ we have $\Lambda_j(x) \to \theta(\star)$ for each $j = 1,2,$ and $\Lambda_3(x) \to 0.$ On the other hand, if $|p(\star)| = 1,$ then $q(\star) = 0,$ and so \cref{equation9: new phase} becomes trivial. We need to check the following two cases separately: $p(\star) = -1$ and $p(\star) = +1.$ If $p(\star) = -1,$ then \cref{equation7: new phase} holds trivially, and \cref{equation8: new phase} follows from $\theta_-(x - 1) = \theta(x - 1)$ and $\theta_+(x) = 0  = \theta(\star),$ where the last equality follows from \cref{equation: limits of theta}.  Similarly, if $p(\star) = +1,$ then \cref{equation8: new phase} holds trivially, and \cref{equation7: new phase} follows from $\theta_+(x) = \theta(x)$ and $\theta_-(x+1) = 0  = \theta(\star).$

(ii) The claim follows from the following analogous equalities:
\begin{align*}
&\lim_{x \to \star} \left( a_+(x) a_+(x - 1) e^{-i(\phi(x) - \phi_+(x) + \phi_-(x - 1))} \right) =
(a(\star) + 1) e^{-i \phi(\star)}, \\
&\lim_{x \to \star} \left( a_-(x) a_-(x + 1) e^{i (\phi(x + 1) - \phi_-(x + 1) + \phi_+(x) )}\right) =
-(a(\star) - 1)e^{i \phi(\star)}, \\
&\lim_{x \to \star}
\left(|b(x)|e^{i(\phi_+(x) - \phi_-(x))}\right)
= |b(\star)|.
\end{align*}
We omit the proof.
\end{proof}

Since the Fredholm index is invariant under multiplication by invertible operators, we have
\[
\ind(e^{-i \theta_+} Q_{\epsilon_0} e^{i \theta_-}) = \ind Q_{\epsilon_0} = \ind(\varGamma,\varGamma'), \qquad
\ind(e^{i \phi_+} Q_{\gamma_0} e^{-i \phi_-}) = \ind Q_{\gamma_0} = \ind(\varGamma',\varGamma).
\]
We are now in a position to apply \cref{theorem: two-phase case}(i) to the strictly local operators $A_\epsilon := e^{-i \theta_+} Q_{\epsilon_0} e^{i \theta_-}$ and $A_\gamma := e^{i \phi_+} Q_{\gamma_0} e^{-i \phi_-}.$ Since the two-sided limits of the coefficients of $-2iA_\epsilon$ and $2iA_\gamma$ are given respectively by \crefrange{equation1: new phase}{equation3: new phase} and \crefrange{equation4: new phase}{equation6: new phase}, we introduce the following two functions for each $\star = \pm \infty$ according to \cref{equation: definition of hatA}:
\begin{align}
\label{equation1: definition of fepsilon}
-2if_\epsilon(z,\star) &:=
(p(\star) + 1) b(\star) e^{i \theta(\star)} z + (p(\star) - 1) b(\star)^* e^{-i \theta(\star)} z^*  -2|q(\star)|a(\star), \qquad z \in \T, \\
\label{equation2: definition of fgamma}
2if_\gamma(z,\star) &:=
(a(\star) + 1) q(\star)^* e^{-i \phi(\star)} z^* + (a(\star) - 1) q(\star) e^{i \phi(\star)} z   -2|b(\star)|p(\star), \qquad z \in \T.
\end{align}
It follows from \cref{theorem: two-phase case}(i) that $A _\epsilon$ is Fredholm (resp. $A _\epsilon$ is Fredholm) if and only if for each $\star = \pm \infty$ the function $f_\epsilon(\cdot,\star)$ is nowhere vanishing (resp. $f_\gamma(\cdot,\star)$ is nowhere vanishing). Moreover,
in this case,
\begin{align}
\label{equation1: witten index expressed as the difference of winding numbers}
\ind(\varGamma,\varGamma') &= \wn(f_\epsilon(\cdot, + \infty)) - \wn(f_\epsilon(\cdot, - \infty)), \\
\label{equation2: witten index expressed as the difference of winding numbers}
\ind(\varGamma',\varGamma) &= \wn(f_\gamma(\cdot, + \infty)) - \wn(f_\gamma(\cdot, - \infty)).
\end{align}
It remains to compute the winding numbers of $f_\epsilon(\cdot, \star), f_\gamma(\cdot, \star)$ by making use of the following elementary fact;

\begin{lemma}
\label{lemma: matsuzawa function is an ellipse}
For each $j=1,2,$ let $(\alpha_j, \beta_j, \Theta_j) \in \R \times \C \times [0,2\pi)$ be a fixed triple satisfying $\alpha_j^2 + |\beta_j|^2 = 1$ and $\beta_j = |\beta_j|e^{i \Theta_j}.$ Let $f : \T \to \C$ be the trigonometric polynomial defined by
\begin{equation}
\label{equation2: definition of matsuzawa function}
2 f(z) :=
(\alpha_1 + 1) \beta_2 e^{i \Theta_1} z + (\alpha_1 - 1) \beta_2^* e^{-i \Theta_1} z^*  -2|\beta_1|\alpha_2,
\qquad z \in \T.
\end{equation}
Then the function $\T \ni z \longmapsto f(z) \in \C$ is nowhere vanishing if and only if $|\alpha_1| \neq |\alpha_2|.$ In this case, we have
\begin{equation}
\label{equation1: winding number of matsuzawa function}
\wn(f) =
\begin{cases}
\sgn \alpha_1, & |\alpha_1| > |\alpha_2|, \\
0,             & |\alpha_1| < |\alpha_2|,
\end{cases}
\end{equation}
where the sign function $\sgn$ is defined by \cref{equation: definition of sign function}.
\end{lemma}
As we shall see below, if $\alpha_1\beta_2 \neq 0,$ then the image of the function $f$ turns out to be an ellipse. In this case, the curve $[0,2\pi] \ni t \longmapsto f(e^{it}) \in \C$ makes precise one revolution around the fixed point $-|\beta_1|\alpha_2$ on the real axis, and so we have $\wn(f) \in \{-1,0,+1\}.$
\begin{proof}
Let us first prove that $f$ is nowhere vanishing if and only if $|\alpha_1 \beta_2| \neq |\beta_1 \alpha_2|.$ In this case,
\begin{equation}
\label{equation2: winding number of matsuzawa function}
\wn(f) =
\begin{cases}
\sgn \alpha_1, & |\alpha_1 \beta_2| >  |\beta_1 \alpha_2|, \\
0,             & |\alpha_1 \beta_2| <  |\beta_1 \alpha_2|.
\end{cases}
\end{equation}
Let us consider the following function on $\R;$
\[
2F(s) := (|\alpha_1 \beta_2| + |\beta_2|) e^{is} + (|\alpha_1 \beta_2| - |\beta_2|)e^{-i s}
= 2 |\alpha_1 \beta_2| \cos s + i 2|\beta_2| \sin s, \qquad s \in \R.
\]
On one hand, if $\alpha_1 \beta_2 = 0,$ then the image of $F$ is a vertical line segment passing through the origin $0 + i0.$ On the other hand, if $\alpha_1 \beta_2 \neq 0,$ then the image of $F$ is an ellipse centred at the origin. For each $t \in [0,2\pi]$
\begin{align*}
2f(e^{i t}) + 2|\beta_1|\alpha_2
&= (\alpha_1 + 1) \beta_2 e^{i \Theta_1} e^{i t} + (\alpha_1 - 1) \beta_2^* e^{-i \Theta_1} e^{-i t} \\
&= (\sgn \alpha_1|\alpha_1| + 1) |\beta_2| e^{i(\Theta_1 + \Theta_2 + t)} + (\sgn \alpha_1|\alpha_1| - 1) |\beta_2| e^{-i(\Theta_1 + \Theta_2 + t)}  \\
&= \sgn \alpha_1 \cdot 2F(\sgn \alpha_1(\Theta_1 + \Theta_2 + t)).
\end{align*}
If $\alpha_1 \beta_2 = 0,$ then the image of the function $[0,2\pi] \ni t \longmapsto f(e^{i t}) \in \C$ coincides with that of the vertical line segment $[-1,1] \ni t \longmapsto -|\beta_1|\alpha_2 + i t|\beta_2| \in \C$ passing through $-|\beta_1|\alpha_2.$ That is, $f$ does not go through the origin if and only if $\beta_1\alpha_2 \neq 0 = \alpha_1 \beta_2,$ and in this case $\wn(f) = 0.$ This is a special case of \cref{equation2: winding number of matsuzawa function}. If $\alpha_1 \beta_2 \neq 0,$ then the image of the curve $[0,2\pi] \ni t \longmapsto f(e^{i t}) \in \C$ is the following ellipse with $\sgn \alpha_1$ being its winding number with respect to the center $-|\beta_1|\alpha_2$ on the real axis;
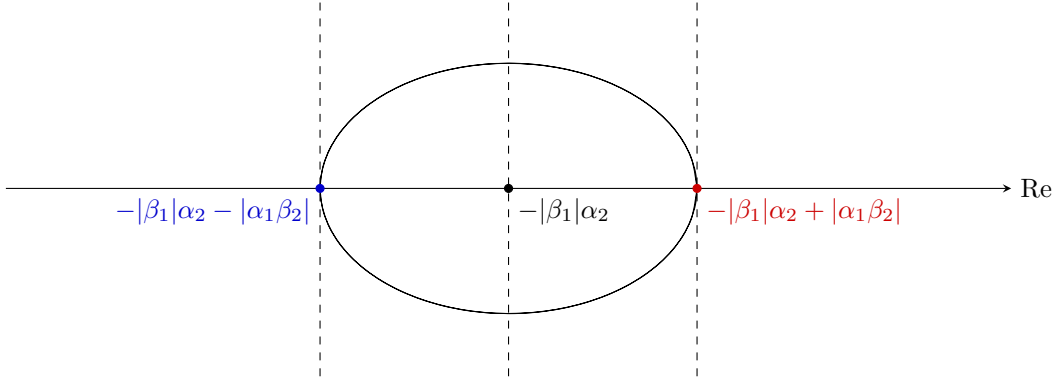
\[
\begin{tikzpicture}
\begin{axis}[axis y line=none,ticks=none,xmin=-8, xmax=8, ymin=-2, ymax=2, legend pos = north west, axis lines=center, xlabel=$\Re$, xlabel style={anchor = west}, width = 0.9\textwidth, height = 0.5\textwidth]
	\addplot [domain=-2*pi:2*pi,samples=50, smooth]({3*cos(deg(x))},{sin(deg(x))});
	\addplot [mark=none,forget plot, dashed] coordinates {(3, -1.5) (3, 1.5)};
	\addplot [mark=none,forget plot, dashed] coordinates {(-3, -1.5) (-3, 1.5)};
	\addplot [mark=none,forget plot, dashed] coordinates {(0, -1.5) (0, 1.5)};
	\draw [fill, black!20!blue] (-3,0) circle (1.5 pt) node [anchor = north east] {$-|\beta_1|\alpha_2 - |\alpha_1 \beta_2|$};
	\draw [fill, black!20!red] (3,0) circle (1.5 pt) node [anchor = north west] {$-|\beta_1|\alpha_2 + |\alpha_1 \beta_2|$};
	\draw [fill] (0,0) circle (1.5 pt) node [anchor = north west] {$-|\beta_1|\alpha_2$};
\end{axis}
\end{tikzpicture}
\]

If $|\alpha_1 \beta_2| >  |\beta_1 \alpha_2|,$ then the origin is inside the interior of the ellipse, and so $\wn(f) = \sgn \alpha_1.$ If $|\alpha_1 \beta_2| <  |\beta_1 \alpha_2|,$ then the origin is inside the exterior of the ellipse, and so $\wn(f) = 0.$ Clearly, the ellipse $f$ goes through the origin if and only if $|\alpha_1 \beta_2| =  |\beta_1 \alpha_2|.$

It remains to check that \cref{equation1: winding number of matsuzawa function} coincides with \cref{equation2: winding number of matsuzawa function}. If the notation $\lessgtr$ simultaneously denotes $>, =, <,$ then
$|\alpha_1 \beta_2| \lessgtr  |\beta_1 \alpha_2|$ if and only if  $|\alpha_1|^2 |\beta_2|^2 \lessgtr  |\beta_1|^2 |\alpha_2|^2$ if and only if $|\alpha_1|^2 \lessgtr |\alpha_2|^2$ if and only if $|\alpha_1| \lessgtr |\alpha_2|.$ The claim follows.
\end{proof}

\begin{proof}[Proof of \cref{theorem: split-step}(i)]
(1) Let $\alpha_1 := p(\star), \beta_1:= q(\star), \alpha_2:= a(\star), \beta_2:= b(\star), \Theta := \theta(\star).$ Then \cref{equation2: definition of matsuzawa function} becomes
\[
2f(z) =
(p(\star) + 1) b(\star) e^{i \theta(\star)} z + (p(\star) - 1)  b(\star)^* e^{-i \theta(\star)} z^*  -2|q(\star)|a(\star) = -2if_\epsilon(z,\star), \qquad z \in \T.
\]
That is, $f = -if_\epsilon(\cdot,\star),$ where the constant $-i$ does not play a significant role in this proof. It follows from \cref{lemma: matsuzawa function is an ellipse} that $f_\epsilon(\cdot,\star)$ is nowhere vanishing if and only if $|p(\star)| \neq |a(\star)|.$ In this case, we have
\begin{equation}
\label{equation: winding number of fepsilon}
\wn(f_\epsilon(\cdot,\star)) = \wn(f) =
\begin{cases}
\sgn p(\star), & |p(\star)| > |a(\star)|, \\
0,             & |p(\star)| < |a(\star)|.
\end{cases}
\end{equation}
The index formula \cref{equation1: Witten index formula} is now an immediate consequence of \cref{equation1: witten index expressed as the difference of winding numbers} and \cref{equation: winding number of fepsilon}.

(2) Let $\alpha_1 := a(\star), \beta_1:= b(\star), \alpha_2:= p(\star), \beta_2:= q(\star), \Theta := \phi(\star).$ Then \cref{equation2: definition of matsuzawa function} becomes
\[
2f(z) =
(a(\star) + 1) q(\star) e^{i \phi(\star)} z + ( a(\star) - 1) q(\star)^* e^{-i \phi(\star)} z^*  -2| b(\star)|p(\star) =
(2if_\gamma(z,\star))^*,
\qquad z \in \T.
\]
That is, $f^* = -if_\epsilon(\cdot,\star).$ It follows from \cref{lemma: matsuzawa function is an ellipse} that $f_\gamma(\cdot,\star)$ is nowhere vanishing if and only if $|a(\star)| \neq |p(\star)|.$ In this case, we have
\begin{equation}
\label{equation: winding number of fgamma}
\wn(f_\gamma(\cdot,\star)) = \wn(f^*) = -\wn(f) =
\begin{cases}
-\sgn a(\star), & |a(\star)| > |p(\star)|, \\
0,             & |a(\star)| < |p(\star)|,
\end{cases}
\end{equation}
where the last equality follows from $\wn(f^*) = -\wn(f).$ The index formula \cref{equation2: Witten index formula} is now an immediate consequence of \cref{equation2: witten index expressed as the difference of winding numbers} and \cref{equation: winding number of fgamma}.
\end{proof}

\cref{theorem: split-step}(i) can also be proved by a purely analytic method without relying on \cref{lemma: matsuzawa function is an ellipse}
(see, for example, \cite[\textsection 4]{Matsuzawa-2020}).

\subsection{Proof of Theorem B (ii)}
\label{section: essential spectrum}

\begin{proof}[Proof of \cref{theorem: split-step} (ii)]
It follows from a direct computation that $U = \varGamma \varGamma'$ is a strictly local operator of the following form;
\begin{align*}
U =
\begin{pmatrix}
p & q L \\
L^{-1} q^*  & -p(\cdot - 1)
\end{pmatrix}
\begin{pmatrix}
a & b^* \\
b & -a
\end{pmatrix}
=
\begin{pmatrix}
q Lb + pa &  - q  La + pb^*  \\
L^{-1}q^*a - p(\cdot - 1) b     &  L^{-1}q^* b^* +  p(\cdot - 1) a
\end{pmatrix}.
\end{align*}
For each $\star = \pm \infty$ and each $z \in \T,$ we introduce the following matrix according to \cref{equation: definition of hatA};
\[
\hat{U}(z, \star)
:=
\begin{pmatrix}
q(\star) b(\star)z + p(\star) a(\star)&  -(q(\star) a(\star) z - p(\star) b(\star)^*)  \\
q(\star)^* a(\star) z^{-1} - p(\star) b(\star)  & q(\star)^* b(\star)^* z^{-1}  + p(\star) a(\star)
\end{pmatrix}
=:
\begin{pmatrix}
X(z,\star) & -Y(z,\star)^* \\
Y(z,\star) & X(z,\star)^*
\end{pmatrix}.
\]
It follows from \cref{theorem: two-phase case}(ii) that the essential spectrum of the evolution operator $U$ is given by
\begin{equation}
\label{equation: definition of sigma pm}
\ess(U) = \sigma(+\infty) \cup \sigma(-\infty), \qquad \sigma(\pm \infty) := \bigcup_{z \in \T} \sigma \left(\hat{U}(z, \pm \infty)\right) =
\bigcup_{t \in [0,2\pi]} \sigma \left(\hat{U}(e^{it}, \pm \infty)\right).
\end{equation}
It remains to compute $\sigma(\star)$ for each $\star = \pm \infty.$ Recall that we have $q(\star) = |q(\star)| e^{i\theta(\star)}$ and $b(\star) = |b(\star)|e^{i\phi(\star)}$ by \cref{equation: limits of theta}. For each $t \in [0,2\pi]$ we have
\begin{align*}
X(e^{it},\star) &= q(\star) b(\star) e^{it} + p(\star) a(\star) = |q(\star) b(\star)| e^{i(t+\theta(\star) + \phi(\star))} + p(\star) a(\star), \\
Y(e^{it},\star) &= q(\star) a(\star) e^{it} - p(\star) b(\star)^* = |q(\star)| a(\star) e^{i(t+\theta(\star))} - p(\star) |b(\star)|e^{-i\phi(\star)}.
\end{align*}
We get the following characteristic equation for each $t \in [0,2\pi];$
\begin{align*}
\det\left(\hat{U}(e^{it}, \star) - \lambda\right)
&= \lambda^2 - (X(e^{it},\star) + X(e^{it},\star)^*)\lambda + X(e^{it},\star)X(e^{it},\star)^* + Y(e^{it},\star)Y(e^{it},\star)^* = 0.
\end{align*}
where $X(e^{it},\star)X(e^{it},\star)^* + Y(e^{it},\star)Y(e^{it},\star)^* = 1$ for each $t \in [0,2\pi]$ by a direct computation. The above characteristic equation becomes the following quadratic equation:
\[
\lambda^2 - 2 \left(p(\star)a(\star)  + |q(\star)  b(\star)| \cos(t + \theta(\star) + \phi(\star))\right) \lambda + 1 = 0, \qquad t \in [0,2\pi].
\]
Thus, if we let $\tau(t, \star) := p(\star) a(\star) + |q(\star) b(\star)| \cos(t + \theta(\star) + \phi(\star))$ for each $t \in [0,2\pi],$ then the above equation has the following two solutions:
\[
\lambda_\pm(t, \star) := \tau(t, \star)  \pm \sqrt{\tau(t, \star)^2 - 1},
\]
where $\{\tau(t,\star)\}_{t \in [0,2\pi]} = [p(\star)a(\star) - |q(\star) b(\star)|, p(\star)a(\star) + |q(\star) b(\star)|] \subseteq [-1,1],$ since
\[
\sup_{t \in [0,2\pi]}
|\tau(t,\star)|
\leq |p(\star)a(\star)| + |q(\star) b(\star)|
\leq \frac{|p(\star)|^2 + |a(\star)|^2}{2} + \frac{|q(\star)|^2 + |b(\star)|^2}{2} = 1.
\]
It follows that
\[
\sigma(\star)
= \bigcup_{t \in [0,2\pi]}\sigma \left(\hat{U}(e^{it}, \star) \right) = \{\lambda_\pm(t, \star)\}_{t \in [0,2\pi]} =  \left\{x \pm i\sqrt{1 - x^2}\right\}_{ x \in [p(\star)a(\star) - |q(\star) b(\star)|, p(\star)a(\star) + |q(\star) b(\star)|]}.
\]
Note that $\pm 1 \in \sigma(\star)$ if and only if $p(\star)a(\star) \pm |q(\star) b(\star)| = \pm 1.$
It can be shown that the last equations are equivalent to $(p(\star) \mp a(\star))^2 = 0,$ and so it follows from \cref{equation: definition of sigma pm} that $\ess(U)$ does not contain both $-1,+1$ if and only if $|p(\star)| \neq |a(\star)|$ for each $\star = \pm \infty.$ Finally, the formula \cref{equation: essential spectrum of U} follows from the fact that $x \in [p(\star)a(\star) - |q(\star) b(\star)|, p(\star)a(\star) + |q(\star) b(\star)|]$ if and only if $\sgn(p(\star)a(\star))x \in I(\star)$ for each $x \in [-1,1].$
\end{proof}

\subsection{Several concluding remarks}
\label{section: concluding remarks}

\subsubsection{The essential spectrum of the imaginary part}

\begin{lemma}
\label{lemma: characterisation of unitary chiral pair}
Let $\cH$ be an abstract Hilbert space, and let $(\varGamma, \varGamma')$ be a chiral pair on $\cH.$ Let $U := \varGamma \varGamma'$ be the associated evolution operator, and let $R,Q$ be the real and imaginary parts of $U$ respectively. Then
\begin{align}
\label{equation: spectral mapping theorem for essential of R}
\ess(R) = \left\{\frac{z + z^*}{2} \mid z \in \ess(U)\right\}, \\
\label{equation: spectral mapping theorem for essential of Q}
\ess(Q) = \left\{\frac{z - z^*}{2i} \mid z \in \ess(U)\right\},
\end{align}
\end{lemma}
\begin{proof}
Let $\cB(\cH) \ni A \longmapsto [A] \in \cB(\cH)/\cK(\cH)$ be the natural surjection onto the Calkin algebra $\cB(\cH)/\cK(\cH).$ If $p : \T \to \C$ is a trigonometric polynomial of the form $p(z) = \sum_{y=-k}^k a(y) z^y$ for each $z \in \T,$ then
\begin{equation}
\ess(p(U))
= \sigma \left(\left[\sum_{y=-k}^k a(y) U^y\right] \right)
= \sigma \left(\sum_{y=-k}^k a(y) \left[U\right]^y \right)
= \sigma \left( p([U]) \right)
= p(\sigma \left( [U] \right))
= p( \ess(U)),
\end{equation}
where the second last equality follows from the spectral mapping theorem. If we let $p(z) := (z + z^*)/2$ (resp. $p(z) := (z - z^*)/(2i)$) for each $z \in \T,$ then $p(U) = R$ (resp. $p(U) = Q$). The claim follows.
\end{proof}

It follows from \cref{equation: spectral mapping theorem for essential of Q} that the following characterisations holds true:
\begin{equation}
\label{equation: characterisation of Fredholmness}
\mbox{The chiral pair } (\varGamma,\varGamma') \mbox{ is Fredholm if and only if } 0 \notin \ess(Q) \mbox{ if and only if } -1,+1 \notin \ess(U).
\end{equation}
Note that \cref{equation: characterisation of Fredholmness} explains why we have the same characterisation $|p(\star)| \neq |a(\star)|$ for each $\star = \pm \infty$ in \cref{theorem: split-step}(i) and \cref{theorem: split-step}(ii).

\begin{theorem}
If $(\varGamma,\varGamma' )$ is the chiral pair in \cref{theorem: split-step} and if $Q$ is the imaginary part of the evolution operator $U := \varGamma \varGamma' ,$ then the following formula holds true;
\begin{equation}
\label{equation: formula for the essential spectrum of Q}
\ess(Q)
= \bigcup_{\star = \pm \infty} \left\{\pm |f_\epsilon(z,\star)|\right\}_{z \in \T}
= \bigcup_{\star = \pm \infty} \left\{\pm |f_\gamma(z,\star)|\right\}_{z \in \T},
\end{equation}
where $f_\epsilon(z,\star)$ and $f_\gamma(z,\star)$ are given respectively by \cref{equation1: definition of fepsilon} and \cref{equation2: definition of fgamma}.
\end{theorem}
\begin{proof}
It follows from \cref{lemma: wada decomposition} that the block-operator matrix representation of $Q_\epsilon := \epsilon^* Q \epsilon$ is given explicitly by the last equality in \cref{equation1: wada decomposition}. Let $\theta_+, \theta_-$ be as in \cref{lemma: phase problem}, and let us consider the following unitary transform of $Q_\epsilon;$
\[
A :=
\begin{pmatrix}
e^{-i \theta_-} & 0 \\
0 & e^{-i \theta_+}
\end{pmatrix}
\begin{pmatrix}
0 & Q_{\epsilon_0}^* \\
Q_{\epsilon_0} & 0
\end{pmatrix}
\begin{pmatrix}
e^{i \theta_-} & 0 \\
0 & e^{i \theta_+}
\end{pmatrix}
=
\begin{pmatrix}
0 & e^{-i \theta_-} Q_{\epsilon_0}^* e^{i \theta_+} \\
e^{-i \theta_+} Q_{\epsilon_0} e^{i \theta_-} & 0
\end{pmatrix}.
\]
We are in a position to apply \cref{theorem: two-phase case}(ii) to the above strictly local operator. We introduce the following matrix-valued function according to \cref{equation: definition of hatA};
\[
\hat{A}(\cdot, \star) :=
\begin{pmatrix}
0 & f_\epsilon(\cdot,\star)^* \\
f_\epsilon(\cdot,\star) & 0
\end{pmatrix}, \qquad \star = \pm \infty,
\]
where $\sigma \left(\hat{A}(z, \star) \right) = \{\pm |f_\epsilon(z,\star)|\}$ for each $z \in \T$ and each $\star = \pm \infty.$ It follows from \cref{theorem: two-phase case}(ii) that
\[
\ess(A)
= \bigcup_{z \in \T} \sigma \left(\hat{A}(\cdot, -\infty)\right)  \cup \bigcup_{z \in \T} \sigma \left(\hat{A}(\cdot, +\infty)\right)
= \bigcup_{z \in \T} \{\pm |f_\epsilon(z,-\infty)|\} \cup \bigcup_{z \in \T} \{\pm |f_\epsilon(z,+\infty)|\}.
\]
Since the essential spectrum is invariant under unitary transforms, we have $\ess(Q) = \ess(Q_\epsilon) = \ess(A).$ It follows that the first equality in \cref{equation: formula for the essential spectrum of Q} holds true. The second equality follows similarly.
\end{proof}

\subsubsection{Topologically protected bound states}
\label{section: topologically protected bound states}

What follows is the subject of another paper in preparation. Let $(\varGamma, \varGamma')$ be a chiral pair defined on an abstract Hilbert space $\cH,$ and let $R, Q$ be the real and imaginary parts of the evolution operator $U := \varGamma \varGamma'$ respectively. With the block-operator matrix representations \crefrange{equation: representation of R}{equation: representation of U} in mind, we define the following two indices:
\begin{equation}
\label{equation: definition of the GW indices}
\ind_\pm (\varGamma, \varGamma') := \dim \ker(R_1 \mp 1) -  \dim \ker(R_2 \mp 1).
\end{equation}
The operator $U = R + iQ$ is unitary, and so $[R,Q] = 0$ and $R^2 + Q^2 = 1.$ It immediately follows from the second equality that the two indices $\ind_\pm (\varGamma, \varGamma')$ given by \cref{equation: definition of the GW indices} are well-defined, if the chiral pair $(\varGamma, \varGamma')$ is Fredholm. In this case, it is not difficult to prove:
\begin{align}
\label{equation: GW indices expressed as the sum and difference of Witten indices}
\ind_\pm (\varGamma, \varGamma') &= \frac{\ind(\varGamma, \varGamma') \pm \ind(\varGamma', \varGamma)}{2}, \\
\label{equation: topologically protected bound states}
|\ind_\pm (\varGamma, \varGamma')| &\leq \dim \ker(U \mp 1),
\end{align}
where non-zero vectors in $\ker(U \mp 1)$ may be referred to as \textbi{topologically protected bound states} as in Physics literature (see, for example, \cite{Kitagawa-Rudner-Berg-Demler-2010,Kitagawa-Broome-Fedrizzi-Rudner-Berg-Kassal-Aspuru-Demler-White-2012}). The three indices $\ind(\varGamma,\varGamma'), \ind_\pm(\varGamma,\varGamma')$ coincide with the \textbi{symmetry indices} $\textrm{si}(U), \textrm{si}_\pm(U)$ discussed in \cite{Cedzich-Geib-Grunbaum-Stahl-Velazquez-Werner-Werner-2018,Cedzich-Geib-Stahl-Velazquez-Werner-Werner-2018}.

If $(\varGamma,\varGamma' )$ is the chiral pair in \cref{theorem: split-step}, then it follows from \cref{theorem: split-step}(i) and \cref{equation: GW indices expressed as the sum and difference of Witten indices} that $\ind_\pm (\varGamma,\varGamma' )$ are given explicitly by the following formulas:
\begin{equation}
\label{equation: GW index for ssqw}
2 \cdot \ind_\pm (\varGamma,\varGamma' ) =
\begin{cases}
\mp \sgn a(+\infty) \pm \sgn a(-\infty),                &  |p(-\infty)| < |a(-\infty)| \mbox{ and } |p(+\infty)| < |a(+\infty)|, \\
+\sgn p(+\infty) \pm \sgn a(-\infty), &  |p(-\infty)| < |a(-\infty)| \mbox{ and }  |p(+\infty)| > |a(+\infty)|, \\
-  \sgn p(-\infty) \mp \sgn a(+\infty), &  |p(-\infty)| > |a(-\infty)| \mbox{ and }  |p(+\infty)| < |a(+\infty)|, \\
+\sgn p(+\infty) - \sgn p(-\infty), &  |p(-\infty)| >  |a(-\infty)| \mbox{ and } |p(+\infty)| > |a(+\infty)|,
\end{cases}
\end{equation}
where we assume $|p(\star)| \neq |a(\star)|$ for each $\star = \pm \infty.$ With \cref{equation: topologically protected bound states} in mind, the formula \cref{equation: GW index for ssqw} can be used to give a lower bound for the number of topologically protected bound states associated with this explicit one-dimensional quantum walk model. This estimate is robust in the sense that the lower bounds $|\ind_\pm (\varGamma,\varGamma' )|$ depend only on the four asymptotic values $p(\pm \infty), a(\pm \infty).$

Note that an estimate of this kind can also be obtained via the \textbi{spectral mapping for discrete-time quantum walks} discussed in \cite{Segawa-Suzuki-2016,Segawa-Suzuki-2019}. Indeed, \cite[Theorem 3.1(iii)]{Suzuki-2019} states that if $U := \varGamma \varGamma'$ denotes the evolution operator of an abstract chiral pair $(\varGamma, \varGamma'),$ where $\varGamma'$ is expressed as $\varGamma' = 2d^*d - 1$ for some fixed co-isometry (i.e. $dd^* = 1$), then the following equalities hold true:
\[
\dim \ker(U \mp 1) = m_\pm + M_\pm,
\]
where $m_\pm := \dim \ker(d \varGamma d^* \mp 1)$ and $M_\pm := \dim \ker(\varGamma \mp 1) \cap \ker d.$ For example, as in \cite[Remark 6.2]{Fuda-Funakawa-Suzuki-2018}, if $U = \varGamma \varGamma' $ is the evolution operator associated with the chiral pair $(\varGamma,\varGamma' ),$ defined in \cref{theorem: split-step}, with $p, q$ being held constant, then one can solve some first-order difference equations to show that there exists a well-defined constant $k > 0,$ such that the condition $|q| < k$ ensures $\dim \ker(U \mp 1) = 1.$

\bibliographystyle{alpha}
\bibliography{Bibliography}

\end{document}